\pdfoutput=1
\documentclass[a4paper,UKenglish,cleveref, autoref, thm-restate]{lipics-v2021}

\usepackage{tikz}
\usepackage{complexity}
\usepackage{comment}
\usepackage[color=blue!20]{todonotes}
\usepackage{xfrac}%
\usepackage{numprint} %
\usepackage[color=blue!20]{todonotes}
\usetikzlibrary{positioning}
\usetikzlibrary { decorations.pathmorphing, decorations.pathreplacing, decorations.shapes, } %
\usetikzlibrary {fit} 
\usepackage{rotating}
\usepackage{subcaption} %

\usepackage[appendix=inline]{apxproof} %

\nolinenumbers

\renewcommand{\mainbodyrepeatedtheorem}{\textup{($\star$)}}
\newtheoremrep{lemma}[theorem]{Lemma}
\newtheoremrep{claim}[theorem]{Claim}
\newtheoremrep{corollary}[theorem]{Corollary}

\usetikzlibrary{patterns}
\tikzstyle{vertex}=[circle,fill=black!0,minimum size=9pt,inner sep=0pt]
\tikzstyle{bigvertex}=[circle,fill=black!0,minimum size=15pt,inner sep=0pt]
\tikzstyle{rect}=[rectangle, rounded corners]
\tikzstyle{edge} = [draw,-,rounded corners=8pt]

\newcommand{\defproblem}[3]{
    \vspace{2mm}
    \vspace{1mm}
    \noindent\fbox{
        \begin{minipage}{0.97\textwidth}
            #1\newline
            {\bf{Input:}} #2\newline
            {\bf{Task:}} #3
        \end{minipage}
    }
    \vspace{2mm}
}

\newcommand{\WHITE}{\texttt{white}}
\newcommand{\BLACK}{\texttt{black}}
\newcommand{\Graph}{G_{\Psi}}
\newcommand{\Ufamily}{\mathcal{S}_\mathcal{U}(G_{\Psi})}
\newcommand{\g}{f^{-1}}

\newcommand{\interval}[5]{%
  \node at (#2+#4/2,#3+0.2) {\scriptsize $I_1(#1)$};
  \draw[color=#5] (#2,#3) node{} -- (#2+#4,#3) node{}; %
  \draw[color=#5] (#2,#3+0.1) -- (#2,#3-0.1);
  \draw[color=#5] (#2+#4,#3+0.1) -- (#2+#4,#3-0.1);
}

\newcommand{\intervalsecond}[5]{%
  \node at (#2+#4/2,#3+0.2) {\scriptsize $I_2(#1)$};
  \draw[color=#5] (#2,#3) node{} -- (#2+#4,#3) node{}; %
  \draw[color=#5] (#2,#3+0.1) -- (#2,#3-0.1);
  \draw[color=#5] (#2+#4,#3+0.1) -- (#2+#4,#3-0.1);
}

\newcommand{\interv}[4]{%
  \node at (#2+0.5,#3+0.2) {\scriptsize $I_1(#1)$};
  \draw[color=#4] (#2,#3) node{} -- (#2+1,#3) node{}; %
  \draw[color=#4] (#2,#3+0.1) -- (#2,#3-0.1);
  \draw[color=#4] (#2+1,#3+0.1) -- (#2+1,#3-0.1);
}

\newcommand{\intervsecond}[4]{%
  \node at (#2+0.5,#3+0.2) {\scriptsize $I_2(#1)$};
  \draw[color=#4] (#2,#3) node{} -- (#2+1,#3) node{}; %
  \draw[color=#4] (#2,#3+0.1) -- (#2,#3-0.1);
  \draw[color=#4] (#2+1,#3+0.1) -- (#2+1,#3-0.1);
}

\renewcommand{\leq}{\leqslant}
\renewcommand{\geq}{\geqslant}

\title{Recognizing unit multiple intervals is hard}

\author{Virginia Ardévol Martínez}{Universit\'{e} Paris-Dauphine, PSL University, CNRS, LAMSADE, 75016 Paris, France}{virginia.ardevol-martinez@dauphine.psl.eu}{https://orcid.org/0000-0002-3703-2335}{}
\author{Romeo Rizzi}{Department of Computer Science, University of Verona, Italy}{romeo.rizzi@univr.it}{https://orcid.org/0000-0002-2387-0952}{}
\author{Florian Sikora}{Universit\'{e} Paris-Dauphine, PSL University, CNRS, LAMSADE, 75016 Paris, France}{florian.sikora@dauphine.fr}{https://orcid.org/0000-0003-2670-6258}{}
\author{Stéphane Vialette}{LIGM, CNRS, Univ Gustave Eiffel, F77454 Marne-la-Vallée, France}{ stephane.vialette@univ-eiffel.fr}{https://orcid.org/0000-0003-2308-6970}{}

\authorrunning{V. Ardévol Martínez, R. Rizzi, F. Sikora, S. Vialette} %

\Copyright{Virginia Ardévol Martínez, Romeo Rizzi, Florian Sikora, Stéphane Vialette} %

\hideLIPIcs



\category{} %

\relatedversion{} %

\acknowledgements{Part of this work was conducted when RR was an invited professor at Université Paris-Dauphine.
This work was partially supported by the  ANR project ANR-21-CE48-0022 (``S-EX-AP-PE-AL'').}

\EventEditors{Satoru Iwata and Naonori Kakimura}
\EventNoEds{2}
\EventLongTitle{34th International Symposium on Algorithms and Computation (ISAAC 2023)}
\EventShortTitle{ISAAC 2023}
\EventAcronym{ISAAC}
\EventYear{2023}
\EventDate{December 3-6, 2023}
\EventLocation{Kyoto, Japan}
\EventLogo{}
\SeriesVolume{283}
\ArticleNo{50}

\begin{document}
\maketitle
\begin{abstract}
Multiple interval graphs are a well-known generalization of interval graphs introduced in the 1970s to deal with situations arising naturally in scheduling and allocation. 
A $d$-interval is the union of $d$ intervals on the real line, and a graph is a $d$-interval graph if it is the intersection graph of $d$-intervals.
In particular, it is a unit $d$-interval graph if it admits a $d$-interval representation where every interval has unit length.

Whereas it has been known for a long time that recognizing 2-interval graphs and other related classes such as 2-track interval graphs is \NP-complete, the complexity of recognizing unit 2-interval graphs remains open.
Here, we settle this question by proving that the recognition of unit 2-interval graphs is also \NP-complete. Our proof technique uses a completely different approach from the other hardness results of recognizing related classes. Furthermore, we extend the result for unit $d$-interval graphs for any $d\geq 2$, which does not follow directly in graph recognition problems --as an example, it took almost 20 years to close the gap between $d=2$ and $d> 2$ for the recognition of $d$-track interval graphs. Our result has several implications, including that recognizing $(x, \dots, x)$ $d$-interval graphs and depth $r$ unit 2-interval graphs is \NP-complete for every $x\geq 11$ and every $r\geq 4$.

\end{abstract}

\section{Introduction}

\emph{Interval graphs} are undirected graphs formed from a set of intervals on the real line, with a vertex for each interval and an edge between vertices whose intervals intersect. In particular, they are chordal and perfect graphs.
Due to its numerous applications
the class of interval graphs is one of the most well-studied classes of graphs~\cite{roberts1978graph,Fishburn1985,mckee1999topics}. 
These include DNA mapping~\cite{zhang1994algorithm},
resource allocation problems in scheduling theory~\cite{bar2001unified} and
ecological niche and food web~\cite{10.2307/j.ctvx5wc04}.

The practical applications of interval graphs have led to the study of various generalizations, including multiple interval graphs \cite{mcguigan1977presentation,DBLP:journals/jgt/TrotterH79,DBLP:journals/siammax/GriggsW80}. A graph is a $d$-interval graph if each vertex is associated with a $d$-interval (the union of $d$ disjoint intervals on the real line) instead of a simple interval, and again, there is an edge between two vertices if and only if the corresponding $d$-intervals overlap at some point of the real line. 
This generalization enables us to model more complex situation arising naturally in scheduling and allocation problems, such as multi-task scheduling, allocation of multiple associated linear resources, or transmission of continuous-media data \cite{Bar-Yehuda2006}.
Applications to bioinformatics, namely to model DNA sequence similarity or RNA secondary structure \cite{Joseph1992,Vialette2004}, increased the interest in this class of graphs. 

Inside the class of multiple interval graphs, different restrictions have been studied. One of the most natural ones is the subclass of unit $d$-interval graphs, which corresponds to $d$-interval graphs that have an interval representation where every interval has unit length.
Unit multiple intervals can be applied, for example, to model tasks of the same duration in scheduling. %

Apart from their concrete applications, another reason why interval graphs have been widely studied in the literature is because many problems that are \NP-hard in general graphs become polynomial-time solvable when restricted to interval graphs: colorability, clique, independent set, or Hamiltonian cycle, to name a few. 
In particular, recognizing interval graphs is also polynomial, and more precisely, it can be done in linear time \cite{Booth1976,corneil2010lbfs}. 
Furthermore, there exist multiple characterizations of interval graphs, including a characterization in terms of forbidden induced subgraphs \cite{lekkeikerker1962representation}.
This is also the case for unit interval graphs \cite{zbMATH03307330}, which are exactly graphs that do not contain any claw, tent, net, or induced cycle of length at least 4. Unit interval graphs are also characterized as interval graphs that are claw-free \cite{roberts1978graph}.

However, for multiple interval graphs, most problems remain hard, even their recognition, and they do not have any simple characterization. In particular, they are neither chordal graphs nor perfect graphs.
It is known that \textsc{Maximum Clique} remains $\NP$-complete in multiple interval graphs, even for unit 2-intervals \cite{DBLP:journals/algorithmica/Francis0O15}, and so do other problems such as \textsc{Independent Set} or \textsc{Dominating Set} \cite{Bar-Yehuda2006,Butman2010}. The parameterized complexity of some of these problems in multiple interval graphs has also been studied, see for instance \cite{DBLP:journals/algorithmica/Jiang13,Fellows2009}.
With respect to the recognition of multiple interval graphs, it was proven to be $\NP$-hard in 1984 \cite{West1984}. More precisely, West and Shmoys showed that determining whether the interval number of a graph (i.e., the smallest integer $d$ such that the graph has a disjoint $d$-interval representation) is smaller or equal to $d$, for any $d\geq 2$, is $\NP$-complete. 
Furthermore, they also proved that for any $r\geq 3$ and any $d\geq 2$, determining whether a graph has an $r$-depth $d$-interval representation (i.e., a $d$-interval representation with at most $r$ intervals sharing a common point) is $\NP$-complete. On the other hand, the complexity of recognizing depth 2 $d$-interval graphs is still open, although it is known to be polynomial for depth 2 unit $d$-interval graphs \cite{DBLP:journals/algorithmica/Jiang13}.
The above-mentioned proof of hardness (for unrestricted depth) was then adapted by Gambette and Vialette for balanced $2$-intervals \cite{Gambette2007}, which are 2-interval graphs that admit a representation such that every 2-interval is composed of two intervals of the same length, while intervals of different 2-intervals can have different lengths. In the same paper, the authors also initiate the study of the recognition of unit 2-interval graphs and of $(x, x)$ 2-interval graphs (where the two disjoint open intervals have integer endpoints and have length $x$), but the complexity of both problems remained unsettled. Note that contrary to the previous characterization by Roberts of unit interval graphs, unit 2-interval graphs \emph{cannot} be characterized as $K_{1,5}$-free 2-interval graphs \cite{alexandre}. %

Another well-studied generalization of interval graphs are $d$-track interval graphs, where each vertex is associated to the \emph{union} of $d$ disjoint intervals, each in a different parallel line called \emph{track}.
Gyárfás and West proved that their recognition is \NP-hard for $d=2$, and conjectured the same for $d\geq 3$~\cite{gyarfas1995multitrack}. 
This conjecture was proven way later in \cite{DBLP:journals/algorithmica/Jiang13} by Jiang, who also showed that recognition remains hard for unit $d$-track interval graphs for any $d \geq 2$, but left the recognition of unit $d$-interval graphs as an open question. %

Multiple track interval graphs can be seen as the union of interval graphs. In the same manner, $d$-boxicity graphs can be seen as the \emph{intersection} of interval graphs. Boxicity is a graph invariant introduced by Roberts \cite{roberts1969recent} and it is the minimum dimension in which a graph can be represented as the intersection graph of boxes. Furthermore, given a graph $G=(V,E)$, it corresponds to the minimum number of interval graphs on the set of vertices $V$ such that the intersection of their edge sets is $G$. Their recognition is $\NP$-complete~\cite{cozzens1982higher,yannakakis1982complexity}, even for $d=2$~\cite{KRATOCHVIL1994233}.

 In this paper, we finally settle the complexity of the recognition of unit 2-interval graphs, answering the open question by  Jiang~\cite{DBLP:journals/algorithmica/Jiang13}. %
 To do so, we prove that it is \NP-hard by reducing from \textsc{Satisfiability} instead of \textsc{Hamiltonian Path}, which has been often used for proving the hardness of the recognition of variants of interval graphs. 
 The reductions from \textsc{Hamiltonian path} in triangle-free cubic graphs used previously to prove the hardness of recognizing $d$-interval graphs, balanced $d$-interval graphs and $d$-track interval graphs all use a special vertex which is adjacent to $n$ vertices of a triangle free graph, and therefore, cannot be directly adapted for unit $2$-interval graphs. We then extend the hardness result for unit $d$-interval graphs, for any $d \geq 2$. Note that, as pointed out in the concluding remarks of~\cite{DBLP:journals/algorithmica/Jiang13}, recognition problems are very different from optimization problems, and the boundary of a graph class is not necessarily harder than that of a subclass\footnote{As an example, the class of $K_{1,5}$-free graphs, which admits a brute-force $\mathcal{O}(n^6)$ time recognition algorithm, contains the class of unit 2-track interval graphs, which is \NP-hard to recognize~\cite{DBLP:journals/algorithmica/Jiang13}.}. Thus, even though one would expect the recognition of unit $d$-interval graphs to be hard for any $d$ if it's hard for $d=2$, it is not directly implied. 
 
 Our result has several consequences, namely that recognizing $(x,\dots ,x)$ $d$-interval graphs and depth $r$ unit $d$-interval graphs is \NP-complete for every $x\geq 11$ and every $r\geq 4$. %
 Finally, our reduction implies as well a lower bound under the ETH.
 
\subparagraph*{Structure of the paper.}

The paper is organized as follows. 
Section \ref{definitions} briefly introduces the necessary concepts and definitions. 
In Section \ref{hardness}, we present the results of the paper. First, in \autoref{coloredunit}, we prove that a generalization of the recognition of unit 2-intervals, \textsc{Colored unit 2-interval recognition}, is \NP-complete. Then, we use this result in \autoref{unit} to prove the main theorem of the paper, which states the \NP-completeness of \textsc{Unit 2-interval recognition}. Finally, we present several implications of our result in \autoref{corollaries}, namely the \NP-completeness of recognizing unit $d$-interval graphs for every $d\geq 2$, and of recognizing $(x,\dots ,x)$ $d$-interval graphs and depth $r$ unit $d$-interval graphs for every $x\geq 11$ and every $r\geq 4$. We conclude with some directions for future work in \autoref{conclusion}.

\section{Definitions}\label{definitions}

An \emph{interval} is a set of real numbers of the form $[a,b] := \{x\in \mathbb{R} \mid a\leq x \leq b\}$.\footnote{In the literature, it is not always specified whether the intervals considered for the intersection representation of interval graphs are open or closed. As discussed in~\cite{rautenbach2013unit}, the reason for this might be that both definitions lead to the same class of finite graphs~\cite{frankl1987open}, even for unit interval graphs. However, note that if we allow the use of both open and closed intervals within one representation, then the class of unit interval graphs obtained is not the same as if we only allowed open or closed intervals within one representation~\cite{rautenbach2013unit}.}

A \emph{$d$-interval} is the union of $d$ disjoint intervals. A $d$-interval is \emph{balanced} if all its $d$ intervals have the same length, and \emph{unit} when this common length is $1$.
A family $\mathcal{F}$ of $d$-intervals is \emph{balanced} (resp., \emph{unit}) if it comprises only balanced (resp., unit) $d$-intervals. Notice that, for $d\geq 2$, different $d$-intervals of a same balanced family may comprise $1$-intervals with different lengths.  
A family $\mathcal{F}$ of $d$-intervals can be used as a representation of the graph $\Omega \, (\mathcal{F}\,)$ having the $d$-intervals of $\mathcal{F}$ as its vertex set, and where two $d$-intervals are adjacent if and only if their intersection is not empty. A graph $G$ is called a (possibly balanced, unit) \emph{$d$-interval graph} when it admits a representation $\mathcal{F}$ consisting only of (respectively balanced, unit) $d$-intervals. Notice that the representing family is not unique (in fact, even only by translating all intervals by a same value, we already obtain an infinite number of them). Multiple interval graphs generalize the standard notion of interval graphs (special case for $d=1$). In this paper, we will use the term \emph{unit 1-interval} (resp. \emph{unit 1-interval graph}) to denote a classical unit interval (resp. a classical unit interval graph), to avoid confusion with a unit 2-interval (resp. unit 2-interval graphs).

Note that many references do not specify whether the intervals of a $d$-interval must be disjoint or not, and some even define them as the union of $d$ not necessarily disjoint intervals \cite{DBLP:journals/jgt/TrotterH79}. However, this might be related to the fact that, when there are no restrictions on the length of the intervals, the two definitions lead to the same class of graphs. This is not true for unit $d$-intervals, so we study the case where disjointness is required, as in the hardness proof of recognizing multiple interval graphs~\cite{West1984}.

A $d$-interval graph is \emph{proper} when it admits a representing family $\mathcal{F}$ such that no $1$-interval is properly contained in another one. The classes of proper and unit $1$-interval graphs are equivalent, and they correspond exactly to $K_{1,3}$-free interval graphs. 
The graph $K_{1,3}$ is the star with 3 leaves, and is also called a \emph{claw}. Equivalently, unit interval graphs are known to be exactly those graphs that do not contain any claw, tent, net, or cycle of length at least 4 as an induced subgraph \cite{zbMATH03307330}.  

A $d$-interval is a $(x_1,\dots,x_d)$ $d$-interval if the $d$ disjoint intervals are open, have integer endpoints, and have lengths $x_1, \dots, x_d$, respectively. 

The \emph{depth} of a family of intervals is the maximum number of intervals that share a common point, and the \emph{representation depth} of a $d$-interval graph is the minimum depth of any $d$-interval representation of the graph.

The hierarchy of subclasses of $d$-interval graphs is as follows \cite{Gambette2007,DBLP:journals/algorithmica/Jiang13}:
$
    (x,\dots, x) \subset (x+1, \dots, x+1) \subset \textit{unit} \subset \textit{balanced} \subset \textit{unrestricted}
$.

The problem \textsc{Unit 2-interval recognition} is defined as follows.

\defproblem{\textsc{Unit 2-interval recognition}}
{A graph $G=(V,E)$}
{Decide whether $G$ has a unit 2-interval representation.
}

Furthermore, we define a more general version of the above problem, which will be useful to prove the hardness of \textsc{Unit 2-interval recognition}.

\defproblem{\textsc{Colored unit 2-interval recognition}}
{A graph $G=(V,E)$ and a coloring $\gamma: V \to \{\WHITE, \BLACK\}$.} %
{Decide whether $G$ has a unit 2-interval representation where:
\begin{itemize}
\item each white vertex is represented by a unit 2-interval,
\item each black vertex is represented by a unit 1-interval.
\end{itemize}
We refer to this representation as a \emph{colored unit 2-interval representation}.

}

\section{Hardness of recognizing unit multiple interval graphs}\label{hardness}
In this section, we prove the main result of this paper, which is the hardness of recognizing unit 2-interval graphs, used later on to prove the hardness of recognizing unit $d$-intervals for every $d\geq 2$. 
The result for $d=2$ is obtained in two steps.
We first prove that the more general version \textsc{Colored unit 2-interval representation} 
is \NP-complete, and then reduce this problem to \textsc{Unit 2-interval recognition}, which yields the main result of this paper.

\subsection{Hardness of \textsc{Colored Unit 2-Interval Recognition}}\label{coloredunit}
Before proceeding to the hardness proof of \textsc{Colored unit 2-interval recognition}, we first introduce the variant of \textsc{SAT} that we will reduce from. In the following, we use the term ``$j$-clause'' to refer to a clause that contains exactly $j$ literals.

\begin{lemmarep}[\cite{DBLP:journals/algorithmica/FellowsKMP95}]
  \label{Restricted Sat}
  \textsc{Satisfiability} is \NP-complete even when restricted to CNF-formulae such that:
  \begin{enumerate}
      \item Every clause contains either $3$ literals ($3$-clause) or $2$ literals ($2$-clause). 
      \item Each variable appears in exactly one $3$-clause.\label{item2}
      \item Each $3$-clause is positive monotone, i.e., is comprised of three positive literals.
      \item Each variable occurs exactly in three clauses, once negated and twice positive.
  \end{enumerate}
\end{lemmarep}
\begin{proof}
This Lemma is proven in \cite[Lemma 2.1]{DBLP:journals/algorithmica/FellowsKMP95}. 
Note that condition~(\ref{item2}) is not explicitly stated in the Lemma's original statement.
However, upon close examination of the proof of Lemma~2.1 given in \cite{{DBLP:journals/algorithmica/FellowsKMP95}}, one can see that condition~(\ref{item2}) holds for all the instances of \textsc{Satisfiability} produced by the proposed reduction if we reduce from an instance of \textsc{3-SAT}. %
Specifically, in the proof, each occurrence of a variable in the original formula is replaced by a new variable, and each new variable (which corresponds to an occurrence of an original variable) also appears in two new 2-clauses. Since the new variable occurs only in these three clauses, it follows that there is exactly one occurrence in a 3-clause if the original instance is an instance form \textsc{3-SAT}. 
\end{proof}

We can now proceed to the proof of hardness of \textsc{Colored Unit $2$-Interval Recognition}. %

\begin{theorem}\label{thm:coloredunit2hard}
\textsc{Colored Unit $2$-Interval Recognition} is \NP-complete, even for graphs of degree at most 6.
\end{theorem}

The rest of the subsection is dedicated to the proof of \autoref{thm:coloredunit2hard}. We first describe the construction used for the reduction and then prove its correctness. 

\subparagraph*{Construction}
Let $\Psi$ be an instance of the variant of \textsc{SAT} described in \autoref{Restricted Sat}, formed by a set of Boolean variables $x_1, \dots , x_n$ and a set of clauses $C_1, \dots , C_m$. %
We construct an equivalent instance $(G_{\Psi}, \gamma_{\Psi})$ of \textsc{Colored unit 2-interval recognition} as follows.

For every variable $x_i$, we introduce the variable gadget $\hat{V_i}$ (truth setting component), which is the vertex-colored graph on three black vertices $A_i$, $B_i$, $C_i$ and three white vertices $x_i^1$, $x_i^2$ and $x_i^N$, with all edges between a black vertex and a white vertex, plus the edges $(x_i^1, x_i^2)$, $(C_i,A_i)$ and $(C_i, B_i)$. We anticipate that the white vertices of $\hat{V_i}$ will be adjacent also to vertices outside $\hat{V_i}$; in order to underline this distinction, these three vertices are called \emph{public}, and the black vertices are called \emph{private}.

        \begin{figure}[h]
        \begin{center}
        \begin{tikzpicture}[transform shape, scale=0.85]
          \node[vertex,draw,white,minimum size=0.7cm,fill=black] (A) at (-1,0.5) {$A_i$};
            \node[vertex,draw,white,minimum size=0.7cm,fill=black] (B) at (-1,-0.5) {$B_i$};
            \node[vertex,draw,white,minimum size=0.7cm,fill=black] (C) at (-2,0) {$C_i$};
            \node[vertex,draw,minimum size=0.7cm,fill=white!40] (x1) at (0,1) {$x_i^1$};
            \node[vertex,draw,minimum size=0.7cm,fill=white!40] (x2) at (0,0) {$x_i^2$};
            \node[vertex,draw,minimum size=0.7cm,fill=white!40] (xn) at (0,-1) {$x_i^N$};
            \draw (C) -- (A) -- (x1) -- (x2) -- (A) -- (xn) -- (B) -- (x2) -- (x1) -- (B) -- (C) -- (x2);
            \draw (xn) edge[bend left] (C) (C) edge [bend left] (x1);
    \end{tikzpicture}
    \caption{Variable gadget $\hat{V_i}$ corresponding to a variable $x_i$. Black vertices are displayed with a black background.}\label{fig:gadget variable}
        \end{center}
        \end{figure}
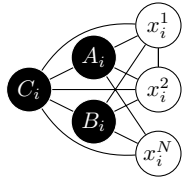

        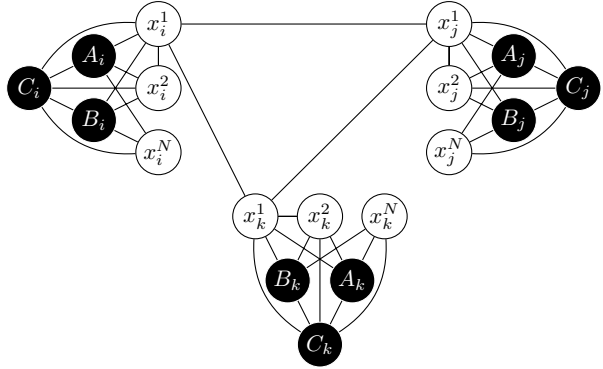
\begin{figure}
        \begin{center}
        \begin{tikzpicture}[transform shape, scale=0.85]
          \node[vertex,draw,white,minimum size=0.7cm,fill=black] (A1) at (-1.5,0.5) {$A_i$};
            \node[vertex,draw,white,minimum size=0.7cm,fill=black] (B1) at (-1.5,-0.5) {$B_i$};
            \node[vertex,draw,white,minimum size=0.7cm,fill=black] (C1) at (-2.5,0) {$C_i$};
            \node[vertex,draw,minimum size=0.7cm,fill=white!40] (x11) at (-0.5,1) {$x_i^1$};
            \node[vertex,draw,minimum size=0.7cm,fill=white!40] (x12) at (-0.5,0) {$x_i^2$};
            \node[vertex,draw,minimum size=0.7cm,fill=white!40] (x1n) at (-0.5,-1) {$x_i^N$};
            \draw (C1) -- (A1) -- (x11) -- (x12) -- (A1) -- (x1n) -- (B1) -- (x12) -- (x11) -- (B1) -- (C1) -- (x12);
            \draw (x1n) edge[bend left] (C1) (C1) edge [bend left] (x11);

            \begin{scope}[xshift=5cm]
            \node[vertex,draw,white,minimum size=0.7cm,fill=black] (A2) at (0,0.5) {$A_j$};
            \node[vertex,draw,white,minimum size=0.7cm,fill=black] (B2) at (0,-0.5) {$B_j$};
            \node[vertex,draw,white,minimum size=0.7cm,fill=black] (C2) at (1,0) {$C_j$};
            \node[vertex,draw,minimum size=0.7cm,fill=white!40] (x21) at (-1,1) {$x_j^1$};
            \node[vertex,draw,minimum size=0.7cm,fill=white!40] (x22) at (-1,0) {$x_j^2$};
            \node[vertex,draw,minimum size=0.7cm,fill=white!40] (x2n) at (-1,-1) {$x_j^N$};
            \draw (C2) -- (A2) -- (x21) -- (x22) -- (A2) -- (x2n) -- (B2) -- (x22) -- (x21) -- (B2) -- (C2) -- (x22);
            \draw (x2n) edge[bend right] (C2) (C2) edge [bend right] (x21);
            \end{scope}

            \node[vertex,draw,white,minimum size=0.7cm,fill=black] (A3) at (2.5,-3) {$A_k$};
            \node[vertex,draw,white,minimum size=0.7cm,fill=black] (B3) at (1.5,-3) {$B_k$};
            \node[vertex,draw,white,minimum size=0.7cm,fill=black] (C3) at (2,-4) {$C_k$};
            \node[vertex,draw,minimum size=0.7cm,fill=white!40] (x31) at (1,-2) {$x_k^1$};
            \node[vertex,draw,minimum size=0.7cm,fill=white!40] (x32) at (2,-2) {$x_k^2$};
            \node[vertex,draw,minimum size=0.7cm,fill=white!40] (x3n) at (3,-2) {$x_k^N$};
            \draw (C3) -- (A3) -- (x31) -- (x32) -- (A3) -- (x3n) -- (B3) -- (x32) -- (x31) -- (B3) -- (C3) -- (x32);
            \draw (x3n) edge[bend left] (C3) (C3) edge [bend left] (x31);

            \draw (x11) -- (x21) -- (x31) -- (x11);
    \end{tikzpicture}
    \caption{Clause gadget $\hat{C_{\alpha}}$ associated to a 3-clause $C_{\alpha}=(x_i \vee x_j \vee x_k)$. Note that in the final graph, each vertex $x_i^m, x_j^m, x_k^m$, for every $m\in\{1,2,N\}$, will be incident to exactly 2 edges linking them to vertices outside their variable gadget.}\label{fig:construction}
        \end{center}
        \end{figure}

\autoref{fig:gadget variable} illustrates the variable gadget $\hat{V_i}$. Notice that the three white node $x_i^1, x_i^2, x_i^N$ correspond each to precisely one of the occurrences of the represented variable $x_i$: vertex $x_i^N$ represents the negated occurrence of $x_i$, vertex $x_i^1$ represents the positive occurrence in a 3-clause, and vertex $x_i^2$ represent the positive occurrence in a 2-clause. Therefore, we refer to them as \emph{literal vertices}.
Furthermore, note  that a vertex of $\hat{V_i}$ is adjacent to $A_i$ if and only if it is adjacent to $B_i$; and being private, these two nodes will remain false twins also in $G$. We will exploit this symmetry to simplify the case analysis. %

To conclude the construction, we show how to encode each clause $C_{\alpha}$, for $\alpha=1,\ldots,m$.
If $C_{\alpha}$ is a $3$-clause, then it is monotone positive, i.e., $C_{\alpha} = ( x_i \vee x_j \vee x_k )$ for some $i,j,k\in \{1,\ldots, n\}$, and all that is needed is to introduce the three edges $(x_i^1, x_j^1)$, $(x_j^1, x_k^1)$, $(x_k^1, x_i^1)$. These three edges comprise the clause gadget (see \autoref{fig:construction}). %

If $C_{\alpha}$ is a 2-clause, say $C_{\alpha}=(x_i^r \vee x_j^s)$ with $i,j\in \{1,\ldots, n\}$ and $ r, s \in \{2,N\}$, then we introduce a public black vertex $L_{i,j}^{\alpha}$ with a private black neighbor $p_{i,j}^{\alpha}$ and we add the four edges $(x_i^r, x_j^s)$, $(x_i^r, L_{i,j}^{\alpha})$, $(x_j^s, L_{i,j}^{\alpha})$ and $(L_{i,j}^{\alpha}, p_{i,j}^{\alpha})$. These four edges together with the two vertices added comprise the clause gadget (see \autoref{fig:construction2clause}).

\begin{figure}[h]
        \begin{center}
        \begin{tikzpicture}[transform shape, scale=0.85]
          \node[vertex,draw,white,minimum size=0.7cm,fill=black] (A1) at (-1,0.5) {$A_i$};
            \node[vertex,draw,white,minimum size=0.7cm,fill=black] (B1) at (-1,-0.5) {$B_i$};
            \node[vertex,draw,white,minimum size=0.7cm,fill=black] (C1) at (-2,0) {$C_i$};
            \node[vertex,draw,minimum size=0.7cm,fill=white!40] (x11) at (0,1) {$x_i^1$};
            \node[vertex,draw,minimum size=0.7cm,fill=white!40] (x12) at (0,0) {$x_i^2$};
            \node[vertex,draw,minimum size=0.7cm,fill=white!40] (x1n) at (0,-1) {$x_i^N$};
            \draw (C1) -- (A1) -- (x11) -- (x12) -- (A1) -- (x1n) -- (B1) -- (x12) -- (x11) -- (B1) -- (C1) -- (x12);
            \draw (x1n) edge[bend left] (C1) (C1) edge [bend left] (x11);

            \begin{scope}[xshift=5cm]
            \node[vertex,draw,white,minimum size=0.7cm,fill=black] (A2) at (0,0.5) {$A_j$};
            \node[vertex,draw,white,minimum size=0.7cm,fill=black] (B2) at (0,-0.5) {$B_j$};
            \node[vertex,draw,white,minimum size=0.7cm,fill=black] (C2) at (1,0) {$C_j$};
            \node[vertex,draw,minimum size=0.7cm,fill=white!40] (x21) at (-1,1) {$x_j^1$};
            \node[vertex,draw,minimum size=0.7cm,fill=white!40] (x22) at (-1,0) {$x_j^2$};
            \node[vertex,draw,minimum size=0.7cm,fill=white!40] (x2n) at (-1,-1) {$x_j^N$};
            \draw (C2) -- (A2) -- (x21) -- (x22) -- (A2) -- (x2n) -- (B2) -- (x22) -- (x21) -- (B2) -- (C2) -- (x22);
            \draw (x2n) edge[bend right] (C2) (C2) edge [bend right] (x21);
            \end{scope}

            \node[vertex,draw,white,minimum size=0.7cm,fill=black] (x32) at (2,-2) {$L_{i,j}^{\alpha}$};
             \node[vertex,draw,white,minimum size=0.7cm,fill=black] (P) at (2,-3) {$p_{i,j}^{\alpha}$};
            
            \draw (x32) -- (P);
            \draw (x12) -- (x2n) -- (x32) -- (x12);
    \end{tikzpicture}
    \caption{Gadget for a $2$-clause $\hat{C_{\alpha}}$ of the form $C_{\alpha}=(x_i \vee \overline{x}_j)$. %
    }\label{fig:construction2clause}
        \end{center}
        \end{figure}
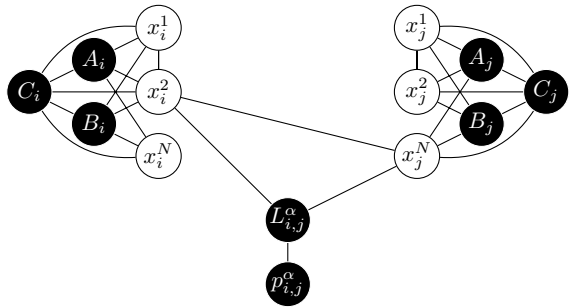

The description of the reduction is complete. Clearly, $G_{\Psi}$ has at most $6n +2m$ vertices and at most $12n + 4m$ edges.
We next introduce a few notions to ease the proof that $G_{\Psi}$ is a colored unit $2$-interval graph if and only if $\Psi$ is satisfiable. 

\begin{definition}\label{definition:split}
Given a colored graph $(G,\gamma)$, we say that a pair $(S,f)$ formed by a graph $S$ and a function $f:V(S)\mapsto V(G)$ is a split of $(G,\gamma)$ if $f$ satisfies the following conditions:
  \begin{itemize}
     \item $|f^{-1}(v)| = 1$ for every $v\in V(G)$ with $\gamma(v)=\BLACK$. 
     \item $|f^{-1}(v)| = 2$ for every $v\in V(G)$ with $\gamma(v)=\WHITE$.
     \item For every vertex $v$ of $G$, $f^{-1}(v)$ is an independent set in $S$.
     \item For every edge $(s,t)$ of $S$, $(f(s),f(t))$ is an edge of $G$.
     \item For every edge $(u,v)$ of $G$, there exist two vertices $s$ and $t$ in $f^{-1}(\{u,v\})$ such that $(s,t)$ is an edge of $S$.
  \end{itemize}
\end{definition}

\begin{definition}
    We define the family of splits of $G$ that lead to a unit 1-interval graph as $\mathcal{S}_\mathcal{U}(G) := \{(S,f) \mid (S,f) \text{ is a split of }$G$ \allowbreak \text{ and $S$ is a unit 1-interval graph} \}$. %
\end{definition}

The next lemma shows how a split $(S,f)$ of a colored graph $G$ can be used to certify that $G$ is a colored unit 2-interval graph. %
This has the advantage of being a truly combinatorial certificate, whereas the number of interval families representing a same graph is infinite with the power of the continuous as soon as at least one exists. 
Trotter and Harary~\cite{DBLP:journals/jgt/TrotterH79} have already studied vertex splitting in the context of turning a graph into an interval graph.%

\begin{lemmarep}\label{lemma:equivalence}
    A colored graph $(G, \gamma)$ is a colored unit 2-interval graph if and only if the family $\mathcal{S}_\mathcal{U}(G)$ is not empty.
\end{lemmarep}
\begin{proof}

Suppose that $G$ is a colored unit $2$-interval graph with $V =V_{\WHITE} \cup  V_{\BLACK}$.
Then, by assumption, there exists a collection of unit $2$-intervals
  $\mathbf{D}_{\WHITE} = \{(I_1(v), I_2(v)) \mid v \in V_{\WHITE}\}$ and a collection of unit intervals
  $\mathbf{I}_{\BLACK} = \{I_1(v) \mid v \in V_{\BLACK}\}$
  such that
  $G \simeq \Omega\left(\mathbf{D}_{\WHITE} \cup \mathbf{I}_{\BLACK}\right)$. 
  Let $\mathcal{F}$ be the family of 1-intervals formed by the ground set of $\mathbf{D}_{\WHITE} \cup \mathbf{I}_{\BLACK}$. Let $S$ be the 1-interval graph defined as the intersection graph of the family $\mathcal{F}$, i.e., $S \simeq \Omega(\mathcal{F})$.
 Consider the function $f:V(S)\mapsto V(G)$ such that: 
 \begin{itemize}
    \item For every $I_1(v) \in \mathbf{D}_{\BLACK}$, $f(I_1(v))= v$.  
     \item For every pair $I_1(v), I_2(v) \in \mathbf{D}_{\WHITE}$, $f(I_1(v))=f(I_2(v)) = v$. 
 \end{itemize}
  By construction, $f$ satisfies all the conditions in \autoref{definition:split}. Indeed, the first three conditions follow directly by definition, while the last two conditions follow because if we have an edge $(I_j(u), I_k(v))$ in $S$, for some $j,k \in \{1,2\}$, this is equivalent to the 2-intervals associated to vertices $u$ and $v$ of $G$ intersecting, so there is an edge $(u,v)$ in $G$.
  Therefore, $(S,f)$ is a split of $(G, \gamma)$.
  
  Conversely, suppose that there exists a split $(S,f)$ of $(G, \gamma)$ that satisfies the property of being a unit interval graph. Then, there exists a collection of unit intervals 
  $\mathbf{I} = \{I_1(s) \mid s \in V(S)\}$
  such that
  $S \simeq \Omega\left(\mathbf{I}\right)$.
  Since $(S,f)$ is a split of $(G,\gamma)$, we know that there exists a map $f:V(S)\mapsto V(G)$ satisfying the conditions in \autoref{definition:split}. We construct a colored unit 2-interval representation of $G$, i.e., a collection of unit 2-intervals $\mathbf{D}_{\WHITE} = \{(I_1(v), I_2(v)) \mid v \in V_{\WHITE}\}$ and a collection of unit 1-intervals
  $\mathbf{I}_{\BLACK} = \{I_1(v) \mid v \in V_{\BLACK}\}$, as follows:
  \begin{itemize}
      \item For every $v\in V(G)$ with $\gamma(v)=\BLACK$, we let $I_1(v)= I_1(s)$, where $s=f^{-1}(v)$.
      \item For every $v\in V(G)$ with $\gamma(v)=\WHITE$, we let $I_1(v)= I_1(s)$ and $I_2(v)=I_1(t)$, where $\{s,t\}=f^{-1}(v)$.
  \end{itemize}
  By construction, this is a colored unit 2-interval representation of $G$, as the last two conditions of $f$ ensure that we preserve the same edges.
\end{proof}

We can now proceed to study the shape of the possible splits $(S,f) \in \Ufamily$.
Let $(S,f)$ be a split of a graph $G$. For every vertex $v \in V(G)$, we call each element of the set $\g(v)$ a representative of $v$. In particular, if $v$ is a white node, we denote its two representatives in $V(S)$ by $\g_1(v)$ and $\g_2(v)$. For simplicity, when we refer to an arbitrary representative of a vertex or to the unique representative of a black vertex, we abuse notation and denote it by its label in $V(G)$. Furthermore, given an edge $(u,v)\in G$, we call the edge $(s,t)\in S$, a representative of $(u,v)$ if $s\in \g(u)$ and $t \in \g(v)$.
Furthermore, given a split $(S,f)$ of the graph $\Graph$, we denote by $S[\hat{V_i}]$ the subgraph of $S$ induced by the vertices of the variable gadget $\hat{V_i}$ (i.e., vertices $A_i, B_i, C_i$, $\g_1(x_i^N), \g_1(x_i^1)$, $\g_1(x_i^2)$, $\g_2(x_i^N), \g_2(x_i^1)$ and $\g_2(x_i^2)$).
Finally, we say that a representative of a literal vertex is an \emph{isolated vertex} if it is not adjacent to any of the private vertices of its variable gadget (i.e., it is not adjacent to $A_i, B_i$ or $C_i$). 

\begin{claim}\label{claim:noduplicates}
 Let $(S,f)$ be an arbitrary graph in $\Ufamily$.
Then, none of the black vertices of $S[\hat{V_i}]$ can be adjacent to both representatives of a literal vertex. Furthermore, if a black vertex is adjacent to a representative of $x_i^1$ and to a representative of $x_i^2$, these two representatives must be adjacent to each other.
\end{claim}
\begin{claimproof}
    Suppose that the two representatives of a literal vertex are adjacent to the same black vertex. If the literal vertex is $x_i^1$ or $x_i^2$, the black vertex would be a center of a $K_{1,3}$ with these two representatives plus a representative of the vertex $x_i^N$ as leaves. If the literal vertex is $x_i^N$, the black vertex would be a center of a $K_{1,3}$ with the two representatives of $x_i^N$ and one of $x_i^1$ or $x_i^2$ as leaves. Since the graph $K_{1,3}$ is a forbidden induced subgraph for unit 1-interval graphs, this contradicts the fact that $S$ belongs to $\Ufamily$.
    Finally, if a black vertex is adjacent to a representative of $x_i^1$ and to a representative of $x_i^2$ which are not adjacent, the black vertex would be a center of a $K_{1,3}$ with these two representatives plus a representative of $x_i^N$ as leaves. 
\end{claimproof}

\begin{claim}\label{splitsneeded}
 Let $(S,f)$ be an arbitrary split in $\Ufamily$.
    Then, %
    for every variable $x_i$ with $i\in\{1,\dots,n\}$, the subgraph $S[\hat{V_i}]$ satisfies at least one of the following two conditions, up to symmetry:
    \begin{enumerate}
        \item\label{splitsneeded:case1} The vertex $\g_1(x_i^N)$ is adjacent to $A_i$ and the vertex $\g_2(x_i^N)$ is adjacent to $B_i$.
        \item\label{splitsneeded:case2} The vertices $\g_1(x_i^1)$ and $\g_1(x_i^2)$ are adjacent to each other and to $A_i$, and the vertices $\g_2(x_i^1)$ and $\g_2(x_i^2)$ are adjacent to each other and to $B_i$. 
    \end{enumerate}
\end{claim}
\begin{claimproof}
    By the properties of $f$, for every edge $(u,v)\in \Graph$, there exist elements $s,t \in V(S)$ with $\g(u)=s$ and $\g(v) =t$ such that $(s,t)$ is an edge in $S$.  
    
    Suppose condition~\ref{splitsneeded:case1} does not hold, i.e., one of the representatives of $x_i^N$, say $\g_1(x_i^N)$, is adjacent to both $A_i$ and $B_i$. We will show that if condition~\ref{splitsneeded:case2} does not hold either, $S$ cannot be a unit 1-interval graph. 
    Assume that one of the representatives of $x_i^1$ or $x_i^2$, say $\g_1(x_i^1)$ (resp. $\g_1(x_i^2)$), is adjacent to both $A_i$ and $B_i$. Then, $S$ contains an induced cycle of length four: $\left(\g_1(x_i^N), B_i, \g_1(x_i^1), A_i\right)$ (resp. $\left(\g_1(x_i^N), B_i, \g_1(x_i^2), A_i\right)$). This is a forbidden induced subgraph for unit 1-interval graphs, so it contradicts the hypothesis.
    Thus, it follows that, up to symmetry, vertices $\g_1(x_i^1)$ and $\g_1(x_i^2)$ need to be adjacent to $A_i$, and vertices $\g_2(x_i^1)$ and $\g_2(x_i^2)$, to $B_i$. 
    Finally, by \autoref{claim:noduplicates}, $\g_1(x_i^1)$ and $\g_1(x_i^2)$ need to be adjacent to each other, so condition~\ref{splitsneeded:case2} must hold.
    
    Conversely, suppose condition~\ref{splitsneeded:case2} does not hold, i.e., at least one of the representatives of $x_i^1$ or $x_i^2$, say $\g_1(x_i^1)$ w.l.o.g., is adjacent to both $A_i$ and $B_i$. We will see that condition~\ref{splitsneeded:case1} must hold in order for $S$ to be a unit 1-interval graph. Indeed, if a single representative of $x_i^N$, say $\g_1(x_i^N)$, is adjacent to both $A_i$ and $B_i$, then $S$ contains an induced cycle of size four: $\left(\g_1(x_i^N), B_i, \g_1(x_i^1), A_i\right)$. Therefore, one representative of $x_i^N$ must be adjacent to $A_i$ and the other, to $B_i$.
\end{claimproof}

  The previous claim implies that there are four possible configuration of $S[\hat{V_i}]$ such that it does not contain any induced cycles of length greater or equal to 4.
  
  \begin{lemma}\label{4options}
  Let $(S,f)$ be a split of $\Graph$ such that $S[\hat{V_i}]$ does not contain any induced cycles of length greater or equal to 4. Then, $S$ satisfies one of the following conditions:
    \begin{enumerate}
        \item The vertex $\g_1(x_i^N)$ is adjacent to $A_i$ and the vertex $\g_2(x_i^N)$ is adjacent to $B_i$, while for the rest of the literal vertices, there exists an element in the image via $\g$ that is an isolated vertex. \label{case1proof}
        
        \item The vertices $\g_1(x_i^1)$ and $\g_1(x_i^2)$ are adjacent to each other and to $A_i$, and the vertices $\g_2(x_i^1)$ and $\g_2(x_i^2)$ are adjacent to each other and to $B_i$, while $\g(x_i^N)$ contains an isolated vertex. \label{case2proof}
        
        \item The images of $x_i^1$ and $x_i^2$ via $\g$ are as in Case~\ref{case1proof} and $\g(x_i^N$) is as in Case~\ref{case2proof} (see the graph in \autoref{fig:split}). \label{case:case3proof}
        
        \item Either the image of $x_i^1$ or the image of $x_i^2$ via $\g$ is as in Case~\ref{case1proof} (w.l.o.g., assume it is $\g(x_i^1)$) so that both representatives of $x_i^1$ are adjacent to the non-isolated representative of $x_i^2$; and $\g(x_i^N)$ is as in Case~\ref{case2proof}. %
        \label{case:case4proof}
    \end{enumerate}
\end{lemma}

\begin{proof}
    We have already shown that one of the conditions of \autoref{splitsneeded} must hold. If condition~\ref{splitsneeded:case1} holds, then we have three possible configurations of $\g(x_i^1)$ and $\g(x_i^2)$: either both literal vertices have a representative that is isolated (Case~\ref{case1proof}), only one of them has a representative that is isolated (Case~\ref{case:case4proof}), or none of them has an isolated representative (Case~\ref{case:case3proof}). On the other hand, if condition~\ref{splitsneeded:case2} holds, the we only have two possible configurations of $\g(x_i^N)$: one representative of $x_i^N$ is isolated (Case~\ref{case2proof}), or none of them is (Case~\ref{case:case3proof}).
    Finally, note that in Case~\ref{case:case4proof}, both representatives of $x_i^1$ need to be adjacent to the non-isolated representative of $x_i^2$ by \autoref{claim:noduplicates}.

\end{proof}

The next two claims are devoted to proving that if $(S,f)$ is a split of $(G_{\Psi}, \gamma)$ contained in the family $\Ufamily$, then Cases~\ref{case:case3proof} and~\ref{case:case4proof} of \autoref{4options} are not possible.
To do so, observe that by construction, since every variable appears exactly in three clauses (twice positive and once negated), we know that in $\Graph$, the vertices $x_i^N, x_i^1$ and $x_i^2$ all have two incident edges linking them with vertices outside of the variable gadget, called \emph{external edges} in the following. 
    The neighbors outside of the variable gadget are \emph{external vertices}, and they constitute private neighbors of the vertices of the variable gadget, as it is not possible for two different vertices of the variable gadget to be incident to the same external neighbor. We will see that if $S$ is as in Case~\ref{case:case3proof} or Case~\ref{case:case4proof}, then the vertices of $S[\hat{V_i}]$ create an induced net with the external neighbors. Since nets are a forbidden induced subgraph for (unit) interval graphs, then $S$ cannot be a unit 1-interval graph. 

\begin{claim}\label{notcase3}
    Let $S$ be an arbitrary graph in $\Ufamily$. Then, for every variable $x_i$ with $i\in\{1,\dots,n\}$, the subgraph $S[\hat{V_i}]$ cannot be as in Case~\ref{case:case3proof} of \autoref{4options}.
\end{claim}

\begin{claimproof}
     Suppose that $S[\hat{V_i}]$ is as in Case~\ref{case:case3proof} of \autoref{4options}, i.e., as in \autoref{fig:split} (where $C_i$ could be in the neighborhood of the other representatives of the vertices, but thanks to the symmetry, these cases are equivalent). We distinguish two cases:
    \begin{itemize}
        \item The two external edges incident to $x_i^1$ and $x_i^2$ are incident to two representatives that are adjacent. Then, either $\g_1(x_i^N), A_i, \g_1(x_i^1), \g_1(x_i^2)$, a private neighbor of $\g_1(x_i^1)$ and a private neighbor of $\g_1(x_i^2)$ form a net; or $\g_2(x_i^N), B_i, \g_2(x_i^1), \g_2(x_i^2)$, a private neighbor of $\g_2(x_i^1)$ and a private neighbor of $\g_2(x_i^2)$ form a net (see the red net in \autoref{fig:split}). 
        \item Otherwise, at least one of $\g_1(x_i^1)$ or $\g_1(x_i^2)$ will be incident to an external edge.
        Then, $C_i, A_i, \g_1(x_i^1)$ or $C_i, A_i, \g_1(x_i^2)$ will create a net together with $B_i, \g_1(x_i^N)$, and the corresponding external neighbor of $\g_1(x_i^1)$ or $\g_1(x_i^2)$, respectively (see the blue net in \autoref{fig:split}). 
    \end{itemize}

        \begin{figure}[h]
        \centering
        \begin{tikzpicture}[transform shape, scale = 0.6]
          \node[vertex,draw,white,minimum size=1.3cm,fill=black] (A) at (-2,1) {$A_i$};
            \node[vertex,draw,white,minimum size=1.3cm,fill=black] (B) at (-2,-1) {$B_i$};
            \node[vertex,draw,white,minimum size=1.3cm,fill=black] (C) at (-4,0) {$C_i$};
            \node[vertex,draw,minimum size=1.3cm,fill=white!40] (x1) at (1,3.5) {$\g_1(x_i^1)$};
            \node[vertex,draw,minimum size=1.3cm,fill=white!40] (y1) at (3,3.5) {$x_l^1$};
            \node[vertex,draw,minimum size=1.3cm,fill=white!40] (x2) at (1,2) {$\g_1(x_i^2)$};
            \node[vertex,draw,minimum size=1.3cm,fill=white!40] (x12) at (1,0.5) {$\g_2(x_i^1)$};
            \node[vertex,draw,minimum size=1.3cm,fill=white!40] (y12) at (3,0.5) {$x_j^1$};
            \node[vertex,draw,minimum size=1.3cm,fill=white!40] (x22) at (1,-1) {$\g_2(x_i^2)$};
            \node[vertex,draw,minimum size=1.3cm,fill=white!40] (y22) at (3,-1) {$x_k^1$};
            \node[vertex,draw,minimum size=1.3cm,fill=white!40] (xn) at (1,-2.5) {$\g_1(x_i^N)$};
            \node[vertex,draw,minimum size=1.3cm,fill=white!40] (xn2) at (1,-4) {$\g_2(x_i^N)$};
            \draw (C) -- (A) -- (x1) -- (x2) -- (A) -- (xn);
            \draw (C) -- (B) -- (x12) -- (x22) -- (B) -- (xn2);
            \draw (xn2) edge[bend left] (C) (C) edge [bend left] (x1);
            \draw (C) edge [bend left] (x2);
            \draw[blue, line width= 0.7mm] (C) edge [bend left] (x1);
            \draw[blue, line width= 0.7mm] (A) -- (x1);
            \draw[red, line width= 0.7mm] (x12) -- (y12);
            \draw[red, line width= 0.7mm] (x22) -- (y22);
            \draw[red, line width= 0.7mm] (xn2)-- (B) -- (x12) --(x22) --(B);
            \draw[blue, line width= 0.7mm] (x1) -- (y1);
            \draw[blue, line width= 0.7mm] (B) -- (C)-- (A) --(xn); 
    \end{tikzpicture}
    \caption{Configuration of $S[\hat{V_i}]$ described in Case~\ref{case:case3proof} of \autoref{4options}. In red, the net created if both $\g_2(x_i^1)$ and $\g_2(x_i^2)$ have an external neighbor. In blue, the net created if $\g_1(x_i^1)$ has an external neighbor.}\label{fig:split}
       
        \end{figure}
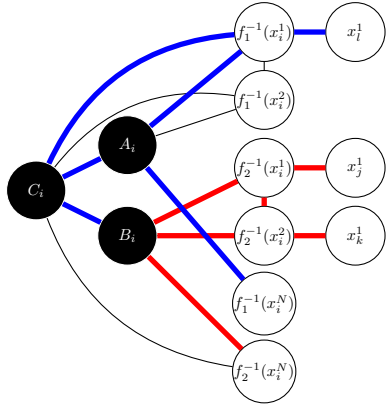

         In both cases, we have a forbidden induced subgraph for (unit) interval graphs, contradicting the hypothesis that $S$ is a unit interval graph.

\end{claimproof}

\begin{claimrep}\label{notcase4}
    Let $(S,f)$ be an arbitrary split in $\Ufamily$. Then, for every variable $x_i$ with $i\in\{1,\dots,n\}$, the subgraph $S[\hat{V_i}]$ cannot be as in Case~\ref{case:case4proof} of \autoref{4options}.
\end{claimrep}
\begin{claimproof}
     Suppose that $S[\hat{V_i}]$ is as in Case~\ref{case:case4proof} of \autoref{4options}, i.e., as in \autoref{fig:split2}. By construction, $x_i^2$ and at least one of $\g_1(x_i^1)$ or $ \g_2(x_i^1)$ have an external neighbor. We distinguish two cases:
    \begin{itemize}
        \item The vertex $\g_1(x_i^1)$ has an external neighbor. Then, vertices $A_i, \g_1(x_i^1), x_i^2, x_i^N$, and the external neighbors of $x_i^2$ and $\g_1(x_i^1)$ form a net (see the red net in \autoref{fig:split2}).
        \item The vertex $\g_2(x_i^1)$ has an external neighbor. Then, vertices $B_i, \g_2(x_i^1), x_i^2, x_i^N$, and the external neighbors of $x_i^2$ and $\g_2(x_i^1)$ form a net (see the blue net in \autoref{fig:split2}).
    \end{itemize}

     \begin{figure}[h]
        \centering
        \begin{tikzpicture}[transform shape, scale=0.6]
          \node[vertex,draw,white,minimum size=1.3cm,fill=black] (A) at (-2,1.5) {$A_i$};
            \node[vertex,draw,white,minimum size=1.3cm,fill=black] (B) at (-2,-0.5) {$B_i$};
            \node[vertex,draw,white,minimum size=1.3cm,fill=black] (C) at (-4,0.75) {$C_i$};
            \node[vertex,draw,minimum size=1.3cm,fill=white!40] (x1) at (1,4) {$\g_1(x_i^1)$};
            \node[vertex,draw,minimum size=1.3cm,fill=white!40] (y1) at (3,4) {$x_j^1$};
            \node[vertex,draw,minimum size=1.3cm,fill=white!40] (y2) at (3,1) {$x_k^1$};
                        \node[vertex,draw,minimum size=1.3cm,fill=white!40] (y3) at (3,-0.5) {$x_l^1$};

            \node[vertex,draw,minimum size=1.3cm,fill=white!40] (x12) at (1,1) {$\g_2(x_i^1)$};
            \node[vertex,draw,minimum size=1.3cm,fill=white!40] (x22) at (1,-0.5) {$x_i^2$};
            \node[vertex,draw,minimum size=1.3cm,fill=white!40] (xn) at (1,-2) {$\g_1(x_i^N)$};
            \node[vertex,draw,minimum size=1.3cm,fill=white!40] (xn2) at (1,-3.5) {$\g_2(x_i^N)$};
            \draw[red, line width= 0.7mm] (x1) -- (y1);
            \draw[blue, line width= 0.7mm] (x12) -- (y2);
            \draw[red, line width= 0.7mm] (x22) -- (y3);
            \draw[blue, line width= 0.7mm, opacity=0.5] (x22) -- (y3);

            \draw[red, line width= 0.7mm] (x22) edge[bend right] (x1);
 \draw[red, line width= 0.7mm] (x22) -- (A);
            \draw[red, line width= 0.7mm] (x1) -- (A) -- (xn) -- (A) -- (x1);
            
            \draw (C) -- (A) -- (x1)  -- (A) -- (xn);
            \draw (C) -- (B);
            \draw[blue, line width= 0.7mm] (B) -- (x12) -- (x22) -- (B) -- (xn2);
            \draw (xn2) edge[bend left] (C) (C) edge [bend left] (x1);
            \draw (A) -- (x22);
    \end{tikzpicture}
    \caption{Configuration of $S[\hat{V_i}]$ described in Case~\ref{case:case4proof} of \autoref{4options}. In red, the net created if $\g_1(x_i^1)$ has an external neighbor, and in blue, the net created if $\g_2(x_i^1)$ has an external neighbor (edge $(x_i^2,x_l^1)$ is part of both nets and is depicted in purple).}\label{fig:split2}
        \end{figure}
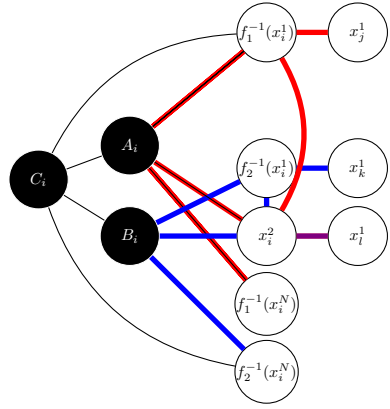

         In both cases, we have a forbidden induced subgraph for (unit) interval graphs, contradicting the hypothesis that $S$ is a unit interval graph.
\end{claimproof}

The proof of \autoref{notcase4} uses similar arguments to that of \autoref{notcase3} and is thus omitted here. Recall that in Case~\ref{case1proof} of \autoref{4options}, one of the representatives of $x_i^1$ and one of the representatives of $x_i^2$ are isolated; and in Case~\ref{case2proof} of \autoref{4options}, one of the representatives of $x_i^N$ is isolated. Therefore, we obtain the following result.

\begin{claim}\label{claimsplitG}
Let $(S,f)$ be an arbitrary split in the family $\Ufamily$. Then, for every variable $x_i$ with $i\in\{1,\dots,n\}$, the subgraph $S[\hat{V_i}]$ satisfies exactly one of the following two conditions: 
\begin{enumerate}
    \item There is a representative of $x_i^1$ and a representative of $x_i^2$ that are isolated vertices (they are either two non-adjacent vertices or they form a $K_2$).\label{case1} 
    \item One of the representatives of $x_i^N$ is an isolated vertex.\label{case2} 
\end{enumerate}
\end{claim}
\begin{claimproof}
Combining \autoref{4options} with Claims~\ref{notcase3} and \ref{notcase4}, 
it follows that $S[\hat{V_i}]$ is either as in Case~\ref{case1proof} or as in Case~\ref{case2proof} of \autoref{4options}, which means that either one representative of each of $x_i^1$ and $x_i^2$ is isolated, or that one representative of $x_i^N$ is isolated, respectively.
These options correspond to the interval representations in \autoref{fig:true2} and \autoref{fig:false2}, respectively. The reader can check the previous assertion observing the figures, and verify that the external edges incident to each of the vertices $x_i^1, x_i^2$ and $x_i^N$ can be added in both representations, as we always have either a whole free interval (not depicted in the figures) or one extreme of the interval free for each of the vertices.
  
\end{claimproof}

 \begin{figure}
\begin{subfigure}[b]{0.45\linewidth}
     \begin{center}
    \begin{tikzpicture}
      \interv{C_i}{0.1}{0.5}{black};
      \interv{A_i}{-0.1}{0}{black};
      \interv{B_i}{1}{0}{black};
      \interv{x_i^N}{-0.7}{-0.5}{black};
      \intervsecond{x_i^N}{1.7}{-0.5}{black};
      \interv{x_i^1}{0.5}{-1}{black};
      \interv{x_i^2}{0.5}{-1.5}{black};
    \end{tikzpicture}
  \end{center}
     \caption{}
     \label{fig:true2}
  \end{subfigure}
  \hfill
  \begin{subfigure}[b]{0.45\linewidth}
     \begin{center}
    \begin{tikzpicture}
      \interv{C_i}{0.1}{0.5}{black};
      \interv{A_i}{-0.1}{0}{black};
      \interv{B_i}{1}{0}{black};
      \interv{x_i^1}{-0.7}{-0.5}{black};
      \intervsecond{x_i^1}{1.6}{-0.5}{black};
      \interv{x_i^2}{-0.6}{-1}{black};
      
      \intervsecond{x_i^2}{1.7}{-1}{black};
      \interv{x_i^N}{0.6}{-1.5}{black};
    \end{tikzpicture}
  \end{center}
    \caption{}
    \label{fig:false2}
  \end{subfigure}
  \caption{Representation of the variable gadget associated to the true value (left, \ref{fig:true2}) or false value (right, \ref{fig:false2}).}
\end{figure}

The correctness of the reduction now follows from the two lemmas below. 

\begin{lemma}\label{directmain}
  If $\Psi$ is satisfiable, then the constructed graph $\Graph=(V,E)$, $V= V_{\WHITE} \cup V_{\BLACK}$, admits a colored unit 2-interval representation.
\end{lemma}

\begin{proof}

Given a satisfying assignment $\phi$ of $\Psi$, we explain how to construct a colored unit 2-interval representation of $\Graph$, i.e., a collection of unit $2$-intervals $\mathbf{D}_{\WHITE} = \{(I_1(v), I_2(v)) \mid v \in V_{\WHITE}\}$ and
  a collection of unit 1-intervals
  $\mathbf{I}_{\BLACK} = \{I_1(v) \mid v \in V_{\BLACK}\}$
  such that
  $G \simeq \Omega\left(\mathbf{D}_{\WHITE} \cup \mathbf{I}_{\BLACK}\right)$.
  Note that by \autoref{lemma:equivalence}, if $\Graph$ is a colored unit 2-interval graph, then there exists a split $(S,f)$ in the family $\Ufamily$, and we know how to construct a colored unit 2-interval representation of $\Graph$ given a unit 1-interval representation of $S$ by defining the 2-interval associated to a white vertex $v\in V_{\WHITE}$ as the union of the interval associated to $\g_1(v)$ and the interval associated to $\g_2(v)$; and the 1-interval associated to a black vertex $v\in V_{\BLACK}$ as the interval associated to the single vertex $\g(v)$.

 For each variable $x_i$ with $ i\in \{1,\dots, n\}$, if $\Phi(x_i) = true$, we represent the variable gadget $\hat{V_i}$
  as shown in \autoref{fig:true2}, which corresponds exactly to Case~\ref{case1} of \autoref{claimsplitG}. 
On the other hand, if $\Phi(x_i) = false$, we represent $\hat{V_i}$ %
as in \autoref{fig:false2}, which corresponds to Case~\ref{case2} of \autoref{claimsplitG}. %
Notice that in both representations, the literals that are true have an isolated representative, i.e., one of the intervals associated to them is unused in the representation of $\hat{V_i}$ and remains completely free to display intersections with external neighbors.

After this, it only remains to explain the connections introduced by the clauses.
\begin{claim}\label{claim3clause}
    Given a 3-clause $(x_i \vee x_j \vee x_k)$, there exists a unit interval representation of the subgraph of $G_{\Psi}$ induced by the vertices of the variable gadgets $\hat{V_i}, \hat{V_j}$ and $\hat{V_k}$. 
\end{claim}

\begin{claimproof}
Each of the variable gadgets can be represented as in \autoref{fig:true2} or \autoref{fig:false2}.
    To represent the edges associated to the 3-clauses, we first notice that, since the 3-clauses are positive monotone, true literals correspond to true variables. As we are assuming that we have a satisfying assignment, we only have three cases (up to symmetry), which correspond to the three variables being true; exactly two variables being true; and only one variable being true.
The literals that are true have a whole free interval to display the intersection, whereas the literals that are false only have the extreme of an interval (while the other extreme is glued to the rest of the representation of the gadget, see \autoref{fig:false2}). 
Let $(x_i \vee x_j \vee x_k)$ be a 3-clause, with $i,j,k \in \{1,\dots,n\}$.
If the three variables are true, we can easily represent the clause by making the three free intervals of the variables -- w.l.o.g. $I_2(x_i^1), I_2(x_j^1), I_2(x_k^1)$ -- intersect at the same time. 
On the other hand, if only one variable -- say $x_i$ -- is false, we can add the two free intervals --$I_2(x_j^1), I_2(x_k^1)$ -- to the corresponding extreme of the gadget of the false variable, as in \autoref{fig:ftt}. 
Finally, if two variables are false -- say $x_i, x_k$ --, then we need to merge the two interval representations associated to their gadgets and add the free interval -- $I_2(x_j^1)$ -- in the middle, as in \autoref{fig:fft}.
Note that the interval representations given in the figures are not unit, but they are proper, so at the end we will be able to use the algorithm described in \cite{DBLP:journals/dm/BogartW99} to turn it into a unit one.
\end{claimproof}

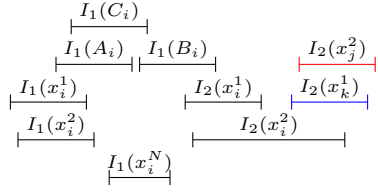
\begin{figure}
     \begin{center}
    \begin{tikzpicture}
      \interv{C_i}{0.1}{0.5}{black};
      \interv{A_i}{-0.1}{0}{black};
      \interv{B_i}{1}{0}{black};
      \interv{x_i^1}{-0.7}{-0.5}{black};
      \intervsecond{x_i^1}{1.6}{-0.5}{black};
      \interv{x_i^2}{-0.6}{-1}{black};
      
      \intervsecond{x_j^2}{3.1}{0}{red};
     \intervsecond{x_k^1}{3}{-0.5}{blue};
      \intervalsecond{x_i^2}{1.7}{-1}{2}{black};

      \interval{x_i^N}{0.6}{-1.5}{0.8}{black};
    \end{tikzpicture}
  \end{center}
    \caption{Representation of a 3-clause $(x_i \vee x_j \vee x_k)$, where $x_i$ is set to false while $x_j, x_k$ are set to true.}
    \label{fig:ftt}
\end{figure}

\begin{figure}
     \begin{center}
    \begin{tikzpicture}
      \interv{C_i}{0.1}{0.5}{black};
      \interv{A_i}{-0.1}{0}{black};
      \interv{B_i}{1}{0}{black};
      \interv{x_i^2}{-0.7}{-0.5}{black};
      \intervsecond{x_i^2}{1.6}{-0.5}{black};
      \interv{x_i^1}{-0.6}{-1}{black};
      
      \intervalsecond{x_j^1}{2.9}{0}{0.6}{red};
     
      \intervalsecond{x_i^1}{1.7}{-1}{1.6}{black};

      \interval{x_i^N}{0.6}{-1.5}{0.8}{black};
      
      \interval{C_k}{4.1}{0.5}{1.1}{blue};
      \interv{A_k}{3.9}{0}{blue};
      \interv{B_k}{5}{0}{blue};
      \interval{x_k^2}{3.8}{-0.5}{0.6}{blue};
      \intervalsecond{x_k^2}{5.7}{-0.5}{1}{blue};
      \interval{x_k^1}{3.2}{-1.5}{1.1}{blue};
      
      \intervalsecond{x_k^1}{5.5}{-1}{1}{blue};

      \interval{x_k^N}{4.7}{-1.5}{0.6}{blue};
    \end{tikzpicture}
  \end{center}
    \caption{Representation of a 3-clause $(x_i \vee x_j \vee x_k)$, where $x_i$ and $x_k$ are set to false and $x_j$ is set to true.}
    \label{fig:fft}
\end{figure}
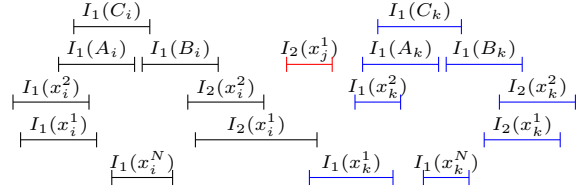

After representing all the 3-clauses, we can assume that the representations of some of the variable gadgets have been merged two by two (we will never have to merge a gadget more than once since a variable occurs in exactly one 3-clause in $\Psi$) and we can fix them in the real line separated from one another. The separation between them can be arbitrarily large, and needs to be at least greater than the space needed to place the remaining intervals. The variable gadgets that have not been merged can also be fixed in the real line, while the unused free intervals (corresponding to true literals), the intervals $I_1(L_{i,j}^{\alpha})$, and the intervals $I_1(p_{i,j}^{\alpha})$ remain unplaced.

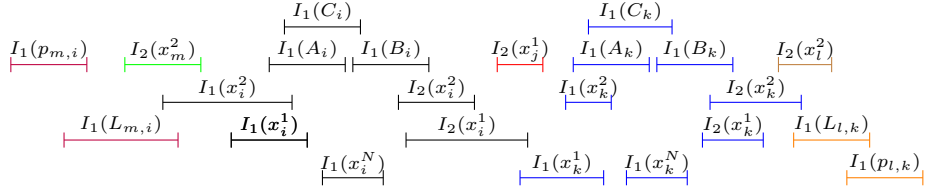
\begin{figure}
     \begin{center}
    \begin{tikzpicture}
      \interv{C_i}{0.1}{0.5}{black};
      \interv{A_i}{-0.1}{0}{black};
      \interv{B_i}{1}{0}{black};
      \interval{x_i^2}{-1.5}{-0.5}{1.7}{black};
      \intervsecond{x_i^2}{1.6}{-0.5}{black};
      \interv{x_i^1}{-0.6}{-1}{black};

      \intervsecond{x_m^2}{-2}{0}{green};
      \interval{L_{m,i}}{-2.8}{-1}{1.5}{purple};
      \interv{p_{m,i}}{-3.5}{0}{purple};

    \intervalsecond{x_l^2}{6.6}{0}{0.7}{brown};
    \interv{L_{l,k}}{6.8}{-1}{orange};
      \interv{p_{l,k}}{7.5}{-1.5}{orange};

      \interv{x_i^1}{-0.6}{-1}{black};
      
      \intervalsecond{x_j^1}{2.9}{0}{0.6}{red};
     
      \intervalsecond{x_i^1}{1.7}{-1}{1.6}{black};

      \interval{x_i^N}{0.6}{-1.5}{0.8}{black};
      
      \interval{C_k}{4.1}{0.5}{1.1}{blue};
      \interv{A_k}{3.9}{0}{blue};
      \interv{B_k}{5}{0}{blue};
      \interval{x_k^2}{3.8}{-0.5}{0.6}{blue};
      \intervalsecond{x_k^2}{5.7}{-0.5}{1.2}{blue};
      \interval{x_k^1}{3.2}{-1.5}{1.1}{blue};
      
      \intervalsecond{x_k^1}{5.6}{-1}{0.8}{blue};

      \interval{x_k^N}{4.6}{-1.5}{0.8}{blue};

    \end{tikzpicture}
  \end{center}
    \caption{Representation of a longest contiguous block of intervals, where each color represents the intervals associated to a different variable. A longest contiguous block occurs when there is a clause $(x_i \vee x_j \vee x_k)$, where $x_i$ and $x_k$ are set to false and both of them also appear as positive literals in a 2-clause.}
    \label{fig:contiguousblock}
\end{figure}

Now, to display the 2-clauses, we distinguish two cases. First, if both literals are true, then there exists a free interval for each, and we can represent the clause in a separate part of the real line (there is one $L_{i,j}^{\alpha}$ and one $p_{i,j}^{\alpha}$ per clause, so these intervals will never cause a problem).
Secondly, if one of the literals is false, then the free interval associated to the true literal needs to be glued to the extreme of the representation of the variable gadget of the false one. 
Note that there is always one free extreme because the 3-clauses use at most one extreme per variable gadget (and we can extend $I_j(x_i^2)$ to allow the intersection while keeping the representation proper).
Note also that we will never need more than two extremes to obtain a representation because, since each variable occurs twice positive and once negated, we can have at most two false literals (when the variable is set to false).

Since we have constructed a proper interval representation, we can now use the algorithm described in~\cite{DBLP:journals/dm/BogartW99} to turn the representation into a unit one, as mentioned before. 

\end{proof}

Let us now prove the converse implication.

\begin{lemma}\label{conversemain}
  If the constructed graph $\Graph=(V,E)$, $V= V_{\WHITE} \cup V_{\BLACK}$, admits a colored unit 2-interval representation, then the original formula $\Psi$ is satisfiable.
\end{lemma}

\begin{proof}

 Assume that the constructed graph $\Graph$ admits a colored unit 2-interval representation where black vertices are represented by unit 1-intervals and white vertices are represented by unit 2-intervals. As in \autoref{claimsplitG}, we study the splits $(S,f)\in \Ufamily$. 
 
 We have already seen in \autoref{claimsplitG} that there are only two possible configurations for $S[\hat{V_i}]$, up to symmetry. %
 Let us assign a truth value to each of the configurations. %
If $S[\hat{V_i}]$ satisfies condition \ref{case1proof} of Claim~\ref{claimsplitG}, we set $\Phi(x_i) = true$. Otherwise, if it satisfies condition \ref{case2proof} of Claim~\ref{claimsplitG}, then we set $\Phi(x_i) = false$. Recall that this implies that there is a representative of the vertices representing true literals which remains isolated from its variable gadget.%

The following claims %
restrict the structure of a representable clause gadget.
Both use similar arguments, so only the first proof is included here.
Given a clause gadget $\hat{C_{\alpha}}$ in $G$, we define the clause gadget $S[\hat{C_{\alpha}}]$ in $S$ as the set of representatives of the edges and vertices of $\hat{C_{\alpha}}$. %

\begin{claim}\label{representable 3clause}
Let $(S,f)$ be an arbitrary split in $\Ufamily$. Then, for every 3-clause, 
at least one of the representatives of the literal vertices incident to the clause gadget in $S$ must be an isolated vertex.
\end{claim}

\begin{claimproof}
    Towards a contradiction, we assume that there exists  a 3-clause gadget in $S$ such that none of the representatives of the literal vertices adjacent to the clause gadget are isolated.
     Let $C_{\alpha}= x_i \vee x_j \vee x_k$, with $i,j,k \in \{1,\dots,n\}$ be a (monotone positive) 3-clause. Each of the literal vertices has two external neighbors. In $S$, either the two external neighbors are incident to the same representative of the literal vertices (and thus only one representative is incident to the clause gadget), or each of them is incident to a different representative. We distinguish two cases, depending on whether only one representative of each literal vertex is incident to the clause gadget, or whether there is at least one literal vertex such that both of its representatives are incident to the clause gadget:
        \begin{itemize}
            \item If only one representative of each literal vertex is incident to the clause gadget in $S$, then w.l.o.g., %
            the clause gadget is formed by edges $\{(\g_1(x_i^1), \g_1(x_j^1)), \allowbreak (\g_1(x_j^1), \g_1(x_k^1)), \allowbreak (\g_1(x_k^1), \g_1(x_i^1)) \}$. By assumption, none of the vertices incident to the clause gadget in $S$ are isolated, so they are all connected to at least one black vertex of their variable gadget. Thus, without loss of generality, $\{\g_1(x_i^1), \g_1(x_j^1), \g_1(x_k^1), A_i, A_j, A_k\}$ form a net (the readers can convince themselves looking at \autoref{fig:construction}). Note that when we say without loss of generality, we are using the symmetry between $A_i$ and $B_i$.  
            \item If there is at least one literal vertex such that both of its representatives are incident to the clause gadget, then w.l.o.g., %
            the clause gadget in $S$ contains edges $\{(\g_1(x_i^1), \g_1(x_j^1)),\allowbreak (\g_1(x_k^1), \g_2(x_i^1)) \}$ (and eventually, edges between representatives of $x_j^1$ and $x_k^1$). %
             Then, since one of the representative of $x_i^2$ also has a private neighbor outside of the variable gadget, either the subgraph induced by $\{A_i, \g_1(x_i^1), \g_1(x_i^2)\}$ or the subgraph induced by $\{B_i, \g_2(x_i^1), \g_2(x_i^2)\}$ (and one private neighbor of each of the three vertices, where the private neighbor of $A_i$ and $B_i$ is $x_i^N$) is a net. This situation is depicted in \autoref{fig:netcase1}.
        
        \end{itemize}

                   \begin{figure}
        \centering
        \begin{tikzpicture}[transform shape, scale=0.6]
          \node[vertex,draw,white,minimum size=1.3cm,fill=black] (A) at (-2,1.5) {$A_i$};
            \node[vertex,draw,white,minimum size=1.3cm,fill=black] (B) at (-2,-0.5) {$B_i$};
            \node[vertex,draw,white,minimum size=1.3cm,fill=black] (C) at (-4,0.75) {$C_i$};
            \node[vertex,draw,minimum size=1.3cm,fill=white!40] (x1) at (1,4) {$\g_1(x_i^1)$};
            \node[vertex,draw,minimum size=1.3cm,fill=white!40] (x2) at (1,2.5) {$\g_1(x_i^2)$};
            \node[vertex,draw,minimum size=1.3cm,fill=white!40] (y1) at (3,4) {$x_j^1$};
            \node[vertex,draw,minimum size=1.3cm,fill=white!40] (y2) at (3,2.5) {$x_k^1$};
            \node[vertex,draw,minimum size=1.3cm,fill=white!40] (x12) at (1,1) {$\g_2(x_i^1)$};
            \node[vertex,draw,minimum size=1.3cm,fill=white!40] (x22) at (1,-0.5) {$\g_2(x_i^2)$};
            \node[vertex,draw,minimum size=1.3cm,fill=white!40] (xn) at (1,-2.5) {$x_i^N$};
            \draw[red, line width= 0.7mm] (x1) -- (y1);
            \draw[red, line width= 0.7mm] (x2) -- (y2);
            \draw[red, line width= 0.7mm] (x1) -- (x2)-- (A) -- (xn) -- (A) -- (x1);
            
            \draw (C) -- (A) -- (x1) -- (x2) -- (A) -- (xn);
            \draw (C) -- (B) -- (x12) -- (x22) -- (B) -- (xn);
            \draw (xn) edge[bend left] (C) (C) edge [bend left] (x1);
            \draw (C) edge [bend left] (x2);
    \end{tikzpicture}
    \caption{In red, the net created if both representatives of $x_i^1$ are incident to the clause gadget in $S$ and $\g_1(x_i^2)$ is incident to an external edge.}\label{fig:netcase1}
        \end{figure}
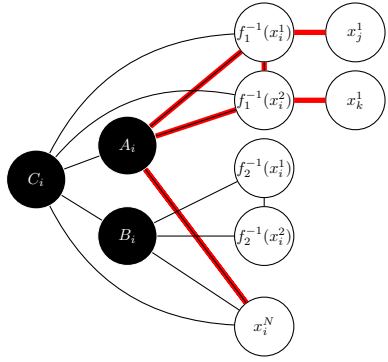

       In both cases, the resulting graph $S$ would not be a unit 1-interval graph, contradicting the hypothesis.
    
\end{claimproof}

\begin{claimrep}\label{representable 2clause}
Let $(S,f)$ be an arbitrary split in $\Ufamily$. Then, for every 2-clause, 
at least one of the representatives of the literal vertices incident to the clause gadget in $S$ must be an isolated vertex. 
\end{claimrep}
\begin{claimproof}
Towards a contradiction, we assume that there exists a 2-clause gadget in $S$ such that none of the representatives of the literal vertices adjacent to the clause gadget are isolated.
 Let $C_{\alpha}= x_i^r \vee x_j^s$, with $i,j \in \{1,\dots,n\}$, be a 2-clause, where the indices $r,s\in \{2, N\}$ indicate which occurrence of the variable appears in the clause. Again, there are two options:
        \begin{itemize}
            \item W.l.o.g., the clause gadget in $S$ %
            comprises edges $\{(\g_1(x_i^r), \g_1(x_j^s)), (\g_1(x_j^s), L_{i,j}^{\alpha}),\allowbreak (L_{i,j}^{\alpha}, \g_1(x_i^r)) \}$. Then, without loss of generality, we will have a net induced by $\{x_i^r, x_j^s, L_{i,j}^{\alpha}, A_i, A_j, p_{i,j}^{\alpha}\}$ (since $L_{i,j}^{\alpha}$ is black and $\g(L_{i,j}^{\alpha})$ consists of a single element, this unique representative will be incident to the clause gadget and adjacent to $p_{i,j}^{\alpha}$ at the same time). The readers can convince themselves looking at figure \autoref{fig:construction2clause}. 
            \item W.l.o.g., the clause gadget in $S$ %
            contains edges $\{(\g_1(x_i^r), \g_1(x_j^s)), (L_{i,j}^{\alpha}, \g_2(x_i^r)) \}$. Suppose first that $x_i^r$ is $x_i^2$. As in the case of 3-clauses, either the subgraph induced by
        $\{A_i, \g_1(x_i^1), \g_1(x_i^2)\}$ or by $\{B_i, \g_2(x_i^1), \g_2(x_i^2)\}$ (and one private neighbor of each of the three vertices, where the private neighbor of $A_i$ and $B_i$ is $x_i^N$) is a net.
        On the other hand, if $x_i^r$ is $x_i^N$, since this literal only occurs in 2-clauses and the vertex $L_{i,j}^{\alpha}$ for 2-clauses is black, then it cannot be the case that $\g_1(x_i^N)$ is adjacent to $x_j^s$, and $\g_2(x_i^N)$ is adjacent to $L_{i,j}^{\alpha}$. Indeed, if this happened, $L_{i,j}^{\alpha}$ would be the center of an induced $K_{1,3}$ with leaves $p_{i,j}^{\alpha}$, $\g_2(x_i^N)$, and a representative of $x_j^s$ (which is not adjacent to $\g_2(x_i^N)$ by assumption). This would contradict the fact that $S$ is a unit 1-interval graph. 
        The illustration of the $K_{1,3}$ created can be seen in \autoref{fig:construction2clause} removing the edge $(x_i^1, x_i^2)$, and replacing $x_i^1$ with $x_i^N$.
        \end{itemize}
         In both cases, the resulting graph $S$ would not be a unit 1-interval graph, contradicting the hypothesis.
\end{claimproof}

 The previous claims imply that there is an isolated literal vertex incident to every 3-clause and to every 2-clause.
Since literal vertices that have an isolated representative correspond to true literals in the assignment fixed before, it follows that there is a true literal per clause, and thus, all clauses are satisfied. This finishes the proof of the converse direction.

\end{proof}

As the problem is clearly in \NP, the polynomial-time construction together with Lemmas \ref{directmain} and \ref{conversemain} conclude the proof of 
\autoref{thm:coloredunit2hard}. The bound on the degree follows because the constructed graph $G$ has maximum degree 6 (the positive literal vertices have degree 4 in the variable gadget and are incident to 2 external edges).

\subsection{Hardness of \textsc{Unit 2-Interval Recognition}}\label{unit}

We show next that \textsc{Colored Unit $2$-Interval Recognition} is polynomial-time reducible to
\textsc{Unit 2-Interval Recognition}, which yields the main result of the paper:

\begin{theorem}
  \label{Color Unit 2-Interval P-reduces 2-Interval}
  \textsc{Unit 2-Interval Recognition} is \NP-complete, even for graphs of degree at most 7.
\end{theorem}

\begin{proof}
  We reduce from \textsc{Colored Unit $2$-Interval Recognition}, which is \NP-hard by \autoref{thm:coloredunit2hard}.
  Given any instance $(G, \gamma)$ of \textsc{Colored Unit $2$-Interval Recognition},
  where $G = (V, E)$ is a graph and
  $\gamma : V \to \{\WHITE, \BLACK\}$ is a vertex-coloring map, we construct an equivalent instance $G' = (V', E')$ of \textsc{Unit 2-Interval Recognition}.
  Define $n = |V|$ and $V_{\texttt{c}} = \{u \mid u \in V \;\wedge\; \gamma(u) = \texttt{c}\}$
  for $c \in \{\WHITE, \BLACK\}$
  (so that $n = \left|V_{\WHITE}\right| + \left|V_{\BLACK}\right|$).
  
  We obtain $G' = (V', E')$ from $G$ by replacing every vertex $v\in V_{\BLACK}$ by the gadget $B_v$ depicted in \autoref{figure:black vertex graph}, which we also call black vertex gadget. Formally, for every $v\in V_{\BLACK}$, we add the vertices $V_{v} = \{ a_v^i, b_v^i \mid  0\leq i \leq 3\} $ and the edges $E_{v} =\{(v, a_{v}^{0}), \allowbreak (a_{v}^{0}, a_{v}^{i}), (v, b_{v}^{0}),\allowbreak (b_{v}^{0}, b_{v}^{i}), \allowbreak (a_{v}^{0}, b_{v}^{0})  \mid 1\leq i \leq 3\}$.
  The gadget $B_v$ is exactly the graph induced by the union of $V_v$ and vertex $v$. 
  Note that the vertex $v$ of $B_v$ is \emph{public}, that is, it is adjacent to vertices of $B_v$ and to vertices outside of $B_v$, while the rest of the vertices of $B_v$ are \emph{private}, i.e., they are only adjacent to vertices of $B_v$.

  We have thus constructed a graph $G'$ with vertex set $V'=V \cup \{V_v \mid v \in V_{\BLACK}\}$ and edge set $E'=E \cup \{E_v \mid v \in V_{\BLACK}\}$. Note that $G'$ contains $G$ as an induced subgraph, as $G'[V]=G$. Combining this with the replacement of every vertex in $V_{\BLACK}$ by a gadget with 9 vertices and 9 edges, it follows that 
  $|V'| = \left|V_{\WHITE}\right| + 9\,\left|V_{\BLACK}\right|$ and
  $|E'| = |E| + 9\,\left|V_{\BLACK}\right|$.

      \begin{figure}[h]
      \centering
      \begin{tikzpicture}
        [transform shape, scale=0.7,
          vertex/.style={draw,circle,fill=black!0,minimum size=6pt,inner sep=0pt}
        ]
        \node [vertex, fill, color =black, label=left:$v$] (v) at (0,0) {};
        \node [vertex,label=left:$a_{v}^{0}$]  (a0) at (-1,-1) {};
        \node [vertex,label=right:$b_{v}^{0}$] (b0) at (1,-1) {};
        \draw (v) -- (a0) -- (b0) -- (v);
        \foreach \i in {1,2,3} {
          \node [vertex,label=below:$a_{v}^{\i}$] (a\i) at (-2 + \i/2,-2) {};
          \draw (a\i) -- (a0);
          \node [vertex,label=below:$b_{v}^{\i}$] (b\i) at (\i/2,-2) {};
          \draw (b\i) -- (b0);
        }
      \end{tikzpicture}
      \caption{\label{figure:black vertex graph}%
        Gadget $B_v$ used to replace every black vertex $v$ of $G$ in the construction of $G'$. Vertex $v$ is a \emph{public} vertex, as it is adjacent to vertices of the gadget ($a_v^0$ and $b_v^0$) and vertices outside the gadget (namely, its neighbors in the original graph $G$), whereas the rest of the vertices are \emph{private}, as their only neighbors are vertices from the gadget (the ones shown in the figure).
      }
    \end{figure}
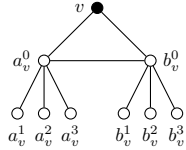

The purpose of the black vertex gadget $B_v$ used to replace every $v\in V_{\BLACK}$ in the construction of $G'$ is to restrict the unit 2-interval representations of $G'$.
Indeed, we will see that it forces one of the intervals associated to $v$ to be used exclusively to represent the gadget, while the other interval is used exclusively to represent the rest of the neighborhood of $v$ (which is exactly its neighborhood in the original graph $G$). 
\autoref{figure:black vertex intervals} shows a unit 2-interval representation $\mathbf{R}= \{(I_1(x),I_2(x))\mid x \in V_v \cup \{v\}\}$ of $B_v$ such that $I_1(v)$ does not have any points in common with the rest of the intervals of $\mathbf{R}$ (i.e., only $I_2(v)$ is used to represent the gadget). Furthermore, in the given representation, $I_2(v)$ cannot intersect any interval associated to a vertex outside of the gadget, as there is no point of $I_2(v)$ that does not intersect either $I_1(a_v^0)$ or $I_1(b_v^0)$, and both $a_v^0$ and $b_v^0$ are private vertices for $v$.
The next claim proves that any unit 2-interval representation of $B_v$ is as in \autoref{figure:black vertex intervals}, up to symmetry. %

 \newcommand{\INTERVALA}[3]{%
    \draw  (#1,#2) -- node [pos=0.5,above] {#3} +(5,0);
    \draw  (#1,#2-0.5) -- (#1,#2+0.5);
    \draw  (#1+5,#2-0.5) -- (#1+5,#2+0.5);
  }
  \newcommand{\INTERVALB}[3]{%
    \draw [very thick] (#1,#2) -- node [pos=0.5,below] {#3} +(5,0);
    \draw [very thick] (#1,#2-0.5) -- (#1,#2+0.5);
    \draw [very thick] (#1+5,#2-0.5) -- (#1+5,#2+0.5);
  }

  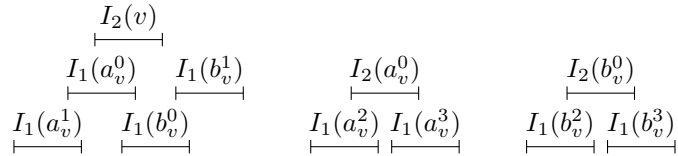
\begin{figure}[h]
    \centering
    \begin{tikzpicture}
      [
        scale=.1775,
      ]
      \INTERVALA{0}{0}{$I_1(a_{v}^{1})$}
      \INTERVALA{4}{4}{$I_1(a_{v}^{0})$}
      \INTERVALA{6}{8}{$I_2(v)$}
      \INTERVALA{8}{0}{$I_1(b_{v}^{0})$}
      \INTERVALA{12}{4}{$I_1(b_{v}^{1})$}
      \begin{scope}[xshift=3cm]
      \INTERVALA{19}{0}{$I_1(a_{v}^{2})$}
      \INTERVALA{22}{4}{$I_2(a_{v}^{0})$}
      \INTERVALA{25}{0}{$I_1(a_{v}^{3})$}
          \begin{scope}[xshift=3cm]
          \INTERVALA{32}{0}{$I_1(b_{v}^{2})$}
          \INTERVALA{35}{4}{$I_2(b_{v}^{0})$}
          \INTERVALA{38}{0}{$I_1(b_{v}^{3})$}
          \end{scope}
      \end{scope}
    \end{tikzpicture}
    \caption{\label{figure:black vertex intervals}%
      A unit 2-interval representation of $B_v$ (\autoref{figure:black vertex graph}), i.e., $\mathbf{D}_{B_v}$ for an arbitrary $v\in V_{\BLACK}$. Note that only one interval of $v$ is used ($I_2(v)$), while the other one remains free to display the rest of the neighborhood of $v$ (and is not represented here). 
    }
  \end{figure}

\begin{claim}\label{claim3}
Let $\{(I_1(x), I_2(x))\mid x \in V_v \cup \{v\}\}$ be a unit 2-interval representation of $B_v$. %
Then, there exist some indices $i,j,k\in \{1,2\}$ such that the representation of $I_i(v), I_j(a_v^0), I_k(b_v^0)$ is contiguous (i.e., the union of the three intervals is an interval) and $I_i(v)$ is properly contained in the union $I_j(a_v^0) \cup I_k(b_v^0)$.
\end{claim}
\begin{claimproof}
    In the following, we denote an interval associated to a vertex by the name of the vertex if it refers to an arbitrary interval from the corresponding 2-interval (i.e., we will write $v$ to denote $I_1(v)$ or $I_2(v)$ when the choice of interval is irrelevant).
    
    Since $a_v^0$ and $b_v^0$ are both centers of an induced $K_{1,4}$, one of the intervals associated to $a_v^0$, say $I_1(a_v^0)$, needs to intersect $v$, $b_v^0$ and one of the $a_v^i$ for some $i\in \{1,2,3\}$, say $a_v^1$ without loss of generality (because of the symmetry). Furthermore, the intervals of $v$ and $b_v^0$ that intersect $I_1(a_v^0)$ also need to intersect each other, as otherwise $I_1(a_v^0)$ would intersect three disjoint intervals, contradicting the fact that the representation is unit. On the other hand, $I_2(a_v^0)$ has to intersect the two remaining $a_v^i$, that is, $a_v^2$ and $a_v^3$. 
    Similarly, one of the intervals associated to $b_v^0$, say $I_1(b_v^0)$, needs to intersect $v$ and $a_v^0$ (which also intersect each other), and one of the $b_v^i$ for some $i\in \{1,2,3\}$, whereas $I_2(b_v^0)$ intersects the two remaining $b_v^i$. 
    Again, without loss of generality, we can assume that $I(b_v^0)$ intersects $b_v^1$ while $I_2(b_v^0)$ intersects $b_v^2$ and $b_v^3$.
    
    Thus, we have that $I_1(a_v^0)$ intersects $v$ and $I_1(b_v^0)$ (which also intersect each other), and $a_v^1$; while $I(b_v^0)$ intersects $v$ and $I_1(a_v ^0)$ (which also intersect each other), and $b_v^1$.
    This implies that the representation of $a_v^1$, $I_1(a_v^0)$, $b_v^1$, $I_1(b_v^0)$ has to be contiguous. 
    Finally, since vertex $v$ is not adjacent to either $a_v^1$ nor $b_v^1$, the only possibility to represent the edges $(v, a_v^0)$ and $(v,b_v^0)$ is by placing an interval associated to $v$, say $I_2(v)$, properly contained in the union $I_1(a_v^0) \cup I_1(b_v^0)$, as in \autoref{figure:black vertex intervals}.
    
\end{claimproof}

The next two claims now prove the correctness of the reduction. 

\begin{claim}\label{colordirect}
    If $G$ is a colored unit $2$-interval graph, then $G'$ is a unit $2$-interval graph.
\end{claim}
\begin{claimproof}
Suppose that $G$ is a colored unit $2$-interval graph.
    Then, by assumption, there exists
  a collection of unit $2$-intervals
  $\mathbf{D}_{\WHITE} = \{(I_1(v), I_2(v)) \mid v \in V_{\WHITE}\}$ and
  a collection of unit intervals
  $\mathbf{I}_{\BLACK} = \{I_1(v) \mid v \in V_{\BLACK}\}$
  such that
  $G \simeq \Omega\left(\mathbf{D}_{\WHITE} \cup \mathbf{I}_{\BLACK}\right)$.
  
  From $\mathbf{D}= \left(\mathbf{D}_{\WHITE} \cup \mathbf{I}_{\BLACK}\right)$,
  we show how to construct a unit $2$-interval representation $\mathbf{D}'$ of $G'$. 
  Recall that $(V_{\WHITE} \cup V_{\BLACK}) = V \subset V'$. Similarly, we will construct $\mathbf{D}'$ such that $\mathbf{D}\subset \mathbf{D}'$. In fact, we will have that $\mathbf{D}' = \mathbf{D} \cup \left(\bigcup_{v\in V_{\BLACK}} \mathbf{D}_{B_v}\right)$, where for every $v\in V _{\BLACK}$, $\mathbf{D}_{B_v}$ is the interval representation of the gadget $B_v$. %
  More precisely, we construct $\mathbf{D}'$ as follows:
  \begin{itemize}
    \item
    For every $v \in V_{\WHITE}$, we add to $\mathbf{D}'$ the $2$-interval
    $(I_1(v), I_2(v))$ from $\mathbf{D}$.
    \item
    For every $v \in V_{\BLACK}$, we add to $\mathbf{D}'$
    the interval $I_1(v)$ from $\mathbf{D}$ together with $\mathbf{D}_{B_v}$, i.e., the interval $I_2(v)$ plus the $2$-intervals
    $(I_1(a_{v}^{k}), I_2(a_{v}^k))$ and $(I_1(b_{v}^{k}), I_2(b_{v}^k))$ for $0 \leq k \leq 3$
    as defined in \autoref{figure:black vertex intervals}.
    
  \end{itemize}
  By construction, $\mathbf{D}'$ is a collection of unit $2$-intervals.
  It is now a simple matter to verify that $G' \simeq \Omega(\mathbf{D}')$.
\end{claimproof}

\begin{claimrep}\label{colorconverse}
     If $G'$ is a unit $2$-interval graph, then $G$ is a colored unit $2$-interval graph.
\end{claimrep}
\begin{claimproof}
    Suppose that $G'$ is a unit $2$-interval graph.
  Then, by assumption, there exists a collection of unit $2$-intervals
  $\mathbf{D'} = \{(I_1(v), I_2(v)) \mid v \in V'\}$ such that
  $G' \simeq \Omega(\mathbf{D'})$.
  From $\mathbf{D'}$, we show how to construct a set $\mathbf{D}$ of $\left|V_{\WHITE}\right|$ unit $2$-intervals and $\left|V_{\BLACK}\right|$ unit intervals. 
  
  Recall that $V \subset V' $. Similarly, we will take $\mathbf{D}$ to be a subset of $\mathbf{D}'$.
  Let $v \in V \subseteq V'$ be a vertex of $G'$.
  We distinguish two cases depending on the color of $v$ in $G$:
  \begin{itemize}
    \item $\gamma(v) = \WHITE$.
    We add to $\mathbf{D}$ the unit $2$-interval
    $(I_1(v), I_2(v))$ of $\mathbf{D}'$.
    \item $\gamma(v) = \BLACK$.
    In $\mathbf{D'}$, we have a pair of intervals $(I_1(v), I_2(v))$. By \autoref{claim3}, w.l.o.g, $I_2(v)$ is used to display the gadget for black vertices, and cannot be used to represent any other edges. This means that all the remaining neighbors of $v$, which are exactly its neighbors in $G$, are displayed by $I_1(v)$.
    Thus, we add to $\mathbf{D}$ the unit interval $I_1(v)$ from $\mathbf{D}'$.
  \end{itemize}
\end{claimproof} 

As the problem is clearly in \NP, combining the fact that the construction of $G'$ can be carried out in polynomial time with Claims \ref{colordirect} and \ref{colorconverse}, we obtain that \textsc{Unit 2-Interval Recognition} is \NP-complete. 
The bound on the degree given in the statement of the theorem follows by construction, from adding the black vertex gadgets (\autoref{figure:black vertex graph}) to the graph constructed in the proof of \autoref{thm:coloredunit2hard} (\autoref{fig:construction}). Indeed, this results in a graph of maximum degree 7, as $C_i$ is adjacent to 5 vertices in the variable gadget and to 2 vertices from the black vertex gadget.

\end{proof}

\subsection{Consequences and generalizations}\label{corollaries}
We now generalize the result for unit $d$-interval graphs, with $d\geq2$, which is not directly implied in graph recognition problems, and for some specific cases of unit $d$-intervals.

\begin{corollaryrep}\label{d-interval}
    Recognizing unit $d$-interval graphs is \NP-complete for every $d\geq 2$.
\end{corollaryrep}

\begin{proof}
    
    We reduce recognition of unit $(d-1)$-interval graphs to recognition of unit $d$-interval graphs, hence the result holds by \autoref{Color Unit 2-Interval P-reduces 2-Interval}.
    
    The idea is similar to the proof of \autoref{Color Unit 2-Interval P-reduces 2-Interval}. Given a graph $G=(V,E)$, we construct in polynomial-time a graph $G'$ by adding to each vertex a gadget similar to the one in \autoref{figure:black vertex graph}. Indeed, for every vertex $v$ in $G$, we create a triangle with vertices $v, a_v^0$ and $b_v^0$, but now $a_v^0$ and $b_v^0$ are adjacent to $2d-1$ independent vertices instead of just 3 (which is the case in \autoref{figure:black vertex graph}). %
    Formally, for every $v\in V$, we add the vertices $$V_{v} = \{ a_v^i, b_v^i \mid  0\leq i \leq 2d-1\} $$ and the edges $$E_{v} =\{(v, a_{v}^{0}), (a_{v}^{0}, a_{v}^{i}), (v, b_{v}^{0}), (b_{v}^{0}, b_{v}^{i}), (a_{v}^{0}, b_{v}^{0})  \mid 1\leq i \leq 2d-1\}$$
    
    We now prove that $G$ has a unit $(d-1)$-interval representation if and only if it has a unit $d$-interval representation.
    First, given a unit $(d-1)$-interval representation, we can build a unit $d$-interval representation as in \autoref{figure:black vertex intervals}.
    However, in this case, instead of having two intervals $I_1(v), I_2(v)$ associated to every vertex $v$, we have $d$ intervals, say $I_1(v), \dots, I_d(v)$. W.l.o.g., the intervals $I_1(v), \dots, I_{d-1}(v)$ are the $d-1$ intervals of the unit $(d-1)$-interval representation of $G$, while $I_d(v)$ plays the role of $I_2(v)$ in \autoref{figure:black vertex intervals}.
    Similarly, $I_1(a_v^0)$ plays the role of $I_1(a_v^0)$ in \autoref{figure:black vertex intervals}, while every $I_j(a_v^0)$, with $1<j\leq d$ is represented as $I_2(a_v^0)$ in \autoref{figure:black vertex intervals}, each intersecting two different $a_v^i$, with $1<i \leq 2d-1$. The same holds for $b_v^0$.
    
    For the converse implication, if we have a unit $d$-interval representation of $G'$, then, using the same argument as in the proof of \autoref{claim3}, we see that for every vertex $v$, we need to use a complete interval of $v$ to represent the gadget added in the construction of $G'$. Therefore, the remaining edges (which correspond exactly to the edges of $G$), need to be displayed using only $d-1$ intervals associated to $v$. This implies that $G$ has a unit $(d-1)$-interval representation. 
\end{proof}

\begin{corollaryrep}\label{corollaryxx}
    Recognizing $(x,\dots, x)$ $d$-interval graphs is \NP-complete for every $x\geq 11$ and every $d\geq 2$.
\end{corollaryrep}
\begin{proof}
    The result follows because the graph constructed in the reduction is a $(11,11)$ 2-interval graph, and every $(11,11)$ 2-interval graph is also a unit 2-interval graph (so the same reduction can be applied). To see this, the reader can verify the $(11,11)$ 2-interval representation of the largest contiguous block in \autoref{fig:5contiguousblock}, and check that the black vertex gadget used in the proof of \autoref{Color Unit 2-Interval P-reduces 2-Interval} is also a $(11,11)$ 2-interval graph.

\newcommand{\openinterval}[5]{%
  \node[minimum size= 2cm] at (#2+11/2,#3+1.5) {  $I_1(#1)$};
  \draw[color=#5] (#2+0.2,#3) node{} -- (#2+11-0.2,#3) node{}; %
}
\newcommand{\openintervalsecond}[5]{%
  \node[minimum size= 2cm] at (#2+11/2,#3+1.5) {  $I_2(#1)$};
  \draw[color=#5] (#2+0.2,#3) node{} -- (#2+11-0.2,#3) node{}; %
}
    \begin{figure}
    \centering
    \begin{sideways}

    \begin{tikzpicture}[ scale=0.17]
      \openinterval{C_i}{40}{10}{5}{black};
      \openinterval{A_i}{30}{5}{5}{black};
      \openinterval{B_i}{41}{5}{5}{black};
      \openinterval{x_i^1}{20}{-1}{5}{black};
      \openintervalsecond{x_i^1}{44}{0}{5}{black};
      \openinterval{x_i^2}{22}{3}{5}{black};

      \openinterval{x_m^1}{11}{5}{5}{green};
      \openinterval{L_{m,i}}{10}{-5}{5}{purple};
      \openinterval{p_{m,i}}{0}{5}{5}{purple};

    \openintervalsecond{x_l^2}{99}{5}{5}{brown};
    \openinterval{L_{l,k}}{100}{-5}{5}{orange};
      \openinterval{p_{l,k}}{110}{-6}{5}{orange};

      \openintervalsecond{x_i^2}{50}{-2}{5}{black};
      
      \openintervalsecond{x_j^2}{55}{5}{5}{red};

      \openinterval{x_i^N}{33}{-5}{5}{black};
      
      \openinterval{C_k}{80}{10}{5}{blue};
      \openinterval{A_k}{70}{5}{5}{blue};
      \openinterval{B_k}{81}{5}{5}{blue};
      \openinterval{x_k^1}{66}{0}{5}{blue};
      \openintervalsecond{x_k^1}{90}{0}{5}{blue};
      \openinterval{x_k^2}{60}{-7}{5}{blue};
      
      \openintervalsecond{x_k^2}{88}{-3}{5}{blue};

      \openinterval{x_k^N}{77}{-4}{5}{blue};
\draw (0,-15) -- (120,-15);

    \foreach \i in {0, 5,...,120} %
    \fill
      [black]
      ( \i,-15)
      circle (3.8 mm)
      node[below=6pt,black] { $\i$ };

    \end{tikzpicture}

    \end{sideways}
 \vspace{0.5cm}
    \caption{An $(11,11)$ 2-interval representation of a longest contiguous block of the graph constructed in the main reduction.}
    \label{fig:5contiguousblock}
\end{figure}
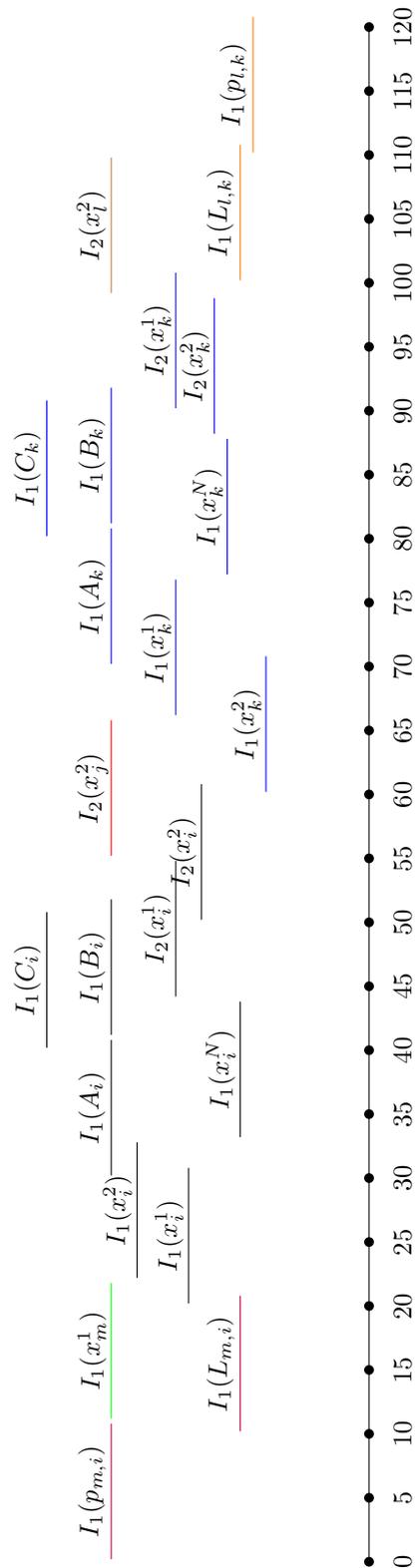

      To generalize for $d>2$, it suffices to check that the gadgets added in the reduction of \autoref{d-interval} are $(11,\dots, 11)$ $d$-interval. Finally, as any $(x, \dots , x)$ $d$-interval graph
can be turned into a $(x + 1, \dots , x + 1)$ $d$-interval graph (by partitioning the $dn$ intervals
into the minimum number of maximal cliques and stretching the intersection of the intervals
in each clique by one unit, as described in~\cite{DBLP:journals/algorithmica/Jiang13}), the graph constructed in the main reduction is an $(x,\dots , x)$ $d$-interval graph for every
$x \geq 11$. However, the graph constructed is not a $(x,\dots , x)$ $d$-interval graph for any
$x < 11$ (this has been checked with the help of an ILP solver). %
\end{proof}

\begin{corollaryrep}\label{cordepth}
     Recognizing depth $r$ unit $d$-interval graphs is \NP-complete for every $r\geq 4$ and every $d\geq 2$.
\end{corollaryrep}
\begin{proof}
    The result follows because the depth of the representation constructed in the hardness proof of \autoref{thm:coloredunit2hard} is 4 (this can be verified by looking at \autoref{fig:contiguousblock}), and the depth of the representation of the black vertex gadget added in the proof of \autoref{Color Unit 2-Interval P-reduces 2-Interval} is 3 (as can be seen in \autoref{figure:black vertex intervals}). Furthermore, the gadgets added to prove the result for $d>2$ have a representation of depth at most 3.
    The corollary generalizes for any depth $r>4$, as for any $r>4$, it is true that there exists a satisfying assignment if and only if the constructed graph $G'$ has a unit 2-interval representation of depth at most $r$.
    
\end{proof}

The following corollary is based on the Exponential Time Hypothesis (ETH). More details on this notion that we are only touching here can be found in~\cite[Chapter 14]{DBLP:books/sp/CyganFKLMPPS15}.

\begin{corollaryrep}\label{corETH}
Unless the ETH fails, \textsc{Unit $d$-interval recognition} does not admit an algorithm with running time $2^{o(|V|+|E|)}$.
\end{corollaryrep}
\begin{proof}
    We have provided a polynomial-time reduction from \textsc{3-SAT} %
    to \textsc{Unit 2-interval recognition} such that given an instance of \textsc{3-SAT} of $n$ variables and $m$ clauses, it outputs an equivalent instance of \textsc{Unit 2-interval recognition} whose size is bounded by $O(n+m)$. 
    Indeed, given an instance of \textsc{3-SAT} of $n$ variables and $m$ clauses, we first build in \autoref{Restricted Sat} an equivalent instance of a special case of \textsc{SAT} with at most $3m$ variables and $7m$ %
    clauses, and then an instance of \textsc{Colored unit 2-interval recognition} with at most $18m$ vertices and $232m$ edges. Finally, to construct an equivalent instance of \textsc{Unit 2-interval recognition}, we also add a linear number of vertices and edges (at most $9|V|$ and $9|E|$, respectively).
    Therefore, if \textsc{Unit 2-interval recognition} admitted an algorithm with running time $2^{o(|V|+|E|)}$, composing the reduction with such an algorithm would yield an algorithm for \textsc{3-SAT} running in time $2^{o(n+m)}$, which would contradict the ETH.
    To generalize the result for $d>2$, notice that the number of vertices and edges that we add in the proof of \autoref{d-interval} is also linear.
\end{proof}

\section{Concluding remarks}\label{conclusion}
We have proven that recognizing unit $d$-interval graphs is \NP-complete for any $d \geq 2$. Furthermore, our reduction implies that recognizing $(x, \dots, x)$ $d$-interval graphs for any $x\geq 11$, and depth $r$ unit $d$-interval graphs for any $r\geq 4$, is also hard. 
These results represent a significant step towards settling the landscape of the complexity of the recognition of the different subclasses of $d$-interval graphs. 

However, some questions still remain open. Since we have shown that recognizing depth 4 unit $d$-interval graphs is \NP-complete and it is known that the recognition of depth 2 unit $d$-interval graphs is polynomial-time solvable~\cite{DBLP:journals/algorithmica/Jiang13}, it still remains to delineate the exact boundary, i.e., study the case of depth 3 unit $d$-interval graphs.
On the other hand, the complexity of recognizing $(x, \dots, x)$ $d$-interval graphs for $x<11$ %
is also unknown.
Finally, we have obtained a lower bound for the running time of an algorithm for recognizing unit 2-intervals. %
Since the brute-force algorithm, running in $\mathcal{O}(2^{n^2})$, is far from achieving it, it would be interesting to reduce this gap.

\bibliographystyle{plainurl}
\bibliography{bib}

\end{document}